\documentclass[10pt,a4paper,twoside,final]{article}

\usepackage[a4paper,
		margin=2cm,
		top=2.5cm,
		bottom=1cm,
		textwidth=183mm,
		includefoot,
		foot=1cm,
		head=2cm]{geometry}
		
\usepackage{amsmath,amsfonts, amssymb, amsthm}

\usepackage{wasysym}
\usepackage{xcolor, framed}
\definecolor{shadecolor}{rgb}{0.9, 0.9, 0.9}

\usepackage[]{graphicx}

\usepackage{multicol}
\usepackage{multirow}

\usepackage{sidecap}
\sidecaptionvpos{figure}{c}
\usepackage[font=small,labelfont=bf,sc]{caption}

\usepackage{siunitx}

\usepackage{enumerate, listings}

\newcommand{\E}[2][]{\ensuremath{\langle #2 \rangle}_{#1}} 
\newcommand{\Var}[1]{\ensuremath{\mathrm{Var}[ #1 ]} }
\newcommand{\Prob}[1]{\ensuremath{\mathrm{P}[ #1 ]}} 
\newcommand{\cond}{\ensuremath{\,| \,}}

\newcommand{\sh}{\ensuremath{\hat m}}
\newcommand{\mh}{\ensuremath{\hat m}}
\newcommand{\convergesinprobability}{\ensuremath{ \xrightarrow[]{\mathrm{p}}}}

\newcommand{\BbbN}{\mathbb{N}}
\newcommand{\BbbR}{\mathbb{R}}

\newsavebox\myboxA
\newsavebox\myboxB
\newlength\mylenA
\newcommand*\widebar[2][0.8]{%
    \sbox{\myboxA}{$\m@th#2$}%
    \setbox\myboxB\null
    \ht\myboxB=\ht\myboxA%
    \dp\myboxB=\dp\myboxA%
    \wd\myboxB=#1\wd\myboxA
    \sbox\myboxB{$\m@th\overline{\copy\myboxB}$}
    \setlength\mylenA{\the\wd\myboxA}
    \addtolength\mylenA{-\the\wd\myboxB}%
    \ifdim\wd\myboxB<\wd\myboxA%
       \rlap{\hskip 0.5\mylenA\usebox\myboxB}{\usebox\myboxA}%
    \else
        \hskip -0.5\mylenA\rlap{\usebox\myboxA}{\hskip 0.5\mylenA\usebox\myboxB}%
    \fi}

\newtheorem{theorem}{\textsc{Theorem}}
\newtheorem{lemma}[theorem]{\textsc{Theorem}}
\newtheorem{corollary}[theorem]{\textsc{Corollary}}
\newtheorem{proposition}[theorem]{\textsc{Conjecture}}
\theoremstyle{definition}
\newtheorem{definition}{\textsc{Definition}}
\theoremstyle{remark}
\newtheorem*{remark}{Remark}

\renewcommand{\thesubsection}{\arabic{subsection}}

\newcommand{\beginsupplement}{%
        \setcounter{table}{0}
        \renewcommand{\thetable}{S\arabic{table}}%
        \setcounter{figure}{0}
        \renewcommand{\thefigure}{S\arabic{figure}}%
        \setcounter{equation}{0}
        \renewcommand{\theequation}{S\arabic{equation}}%
        \setcounter{theorem}{0}
        \renewcommand{\thetheorem}{S\arabic{theorem}}%
        \setcounter{definition}{0}
        \renewcommand{\thedefinition}{S\arabic{definition}}%
        \renewcommand{\thesubsection}{Supp. \arabic{subsection}}
     }

\newcommand{\autocite}{\cite}
\newcommand{\textcite}{\cite}

\title{Inferring collective dynamical states from widely unobserved systems}

\author{J. Wilting$^1$ and V. Priesemann$^{1,2}$ \\ {\small $^1$Max-Planck-Institute for Dynamics and Self-Organization, G\"ottingen} \\ {\small $^2$Bernstein-Center for Computational Neuroscience, G\"ottingen}}
\date{ \normalsize \today}

\begin{document}

\twocolumn[
  \begin{@twocolumnfalse}

\maketitle
       
\vspace*{0.1cm}

\noindent
\textbf{
\noindent
When assessing spatially-extended complex systems, one can rarely sample the states of all components.
We show that this spatial subsampling typically leads to severe underestimation of the risk of instability in systems with propagating events.
We derive a subsampling-invariant estimator, and demonstrate that it correctly infers the infectiousness of various diseases under subsampling, making it particularly useful in countries with unreliable case reports.
In neuroscience, recordings are strongly limited by subsampling.
Here, the subsampling-invariant estimator allows to revisit two prominent hypotheses about the brain’s collective spiking dynamics: asynchronous-irregular or critical.
We identify consistently for rat, cat and monkey a state that combines features of both and allows input to reverberate in the network for hundreds of milliseconds.
Overall, owing to its ready applicability, the novel estimator paves the way to novel insight for the study of spatially-extended dynamical systems.
}
\vspace*{1cm}

\vspace*{0.1cm}

  \end{@twocolumnfalse}
]




How can we infer properties of a high-dimensional dynamical system if we can only observe a very small part of it? This problem of spatial subsampling is common to almost every area of research where spatially extended, time evolving systems are investigated.
For example, in many diseases the number of reported infections may be much lower than the unreported ones \autocite{Papoz1996}, or in the financial system only a subset of all banks is evaluated when assessing the risk of developing system wide instability \autocite{Quagliariello2009} (``stress test'').
Spatial subsampling is particularly severe when recording neuronal spiking activity, because the number of neurons that can be recorded with ms precision is vanishingly small compared to the number of all neurons in a brain area \autocite{Priesemann2009,Ribeiro2010,Ribeiro2014,Levina2017} (Fig. \ref{fig:branching_cartoon}\textbf{a}).

Here, we show that subsampling leads to a strong overestimation of stability in a large class of time evolving systems (\ref{sec:supp_applicability}), which include epidemic spread of infectious diseases  \autocite{Farrington2003}, cell proliferation, evolution (see \cite{Kimmel2015} and references therein), neutron processes in nuclear power reactors \autocite{Pazy1973}, spread of bank-ruptcy \autocite{Filimonov2012}, evolution of stock prices \autocite{Mitov2009}, or the propagation of spiking activity in neural networks \autocite{Beggs2003,Haldeman2005} (Fig. \ref{fig:branching_cartoon}\textbf{b}).
However, correct risk prediction is essential to timely initiate counter actions to mitigate the propagation of events.
We introduce a novel estimator that allows correct risk assessment even under strong subsampling.
Mathematically, the evolution of all these systems is often approximated by a process with a 1$^{\mathrm{st}}$ order autoregressive representation (PAR), e.g. by an AR(1), branching, or Kesten process (Fig. \ref{fig:supp_ar_processes}, \ref{sec:supp_bps}).
For these processes, we derive first the origin of the estimation bias and develop a novel estimator, which we analytically prove to be consistent under subsampling.
We then apply the novel estimator to models and real-world data of disease and brain activity.
To assure that a PAR is a reasonable approximation of the complex system under study, and to exclude contamination through potential non-stationarities, we included a set of automated, data-driven tests.

In a PAR\footnotemark[2], the activity in the next time step, $A_{t + 1}$, depends linearly on the current activity $A_t$.
In addition, it incorporates external input, e.g. drive from stimuli or other brain areas, with a mean rate $h$, yielding the autoregressive representation 

\footnotetext[2]{
For the mathematically inclined reader we recommend the detailed derivation in \ref{sec:supp_applicability} -- 4.} 

\begin{equation}
\E{A_{t+1}|A_{t}} = m\, A_{t} + h,
\label{eq:autoreg}
\end{equation}

\noindent
where $\E{\cdot \cond \cdot}$ denotes the conditional expectation.
The stability of $A_{t}$ is solely governed by $m$, e.g. the mean number of persons infected by \emph{one} diseased person \autocite{Heathcote1965}.
The activity is stationary if $m<1$, while it grows exponentially if $m>1$.
The state $m=1$ separates the stable from the unstable regime. 
Especially close to this transition, a correct estimate of $m$ is vital to assess the risk that $A_{t}$ develops a large, potentially devastating cascade or avalanche of events (e.g. an epidemic disease outbreak or an epileptic seizure), either generically or via a minor increase in $m$. 

\begin{figure*}[]
\centering
\includegraphics[width=0.9\textwidth]{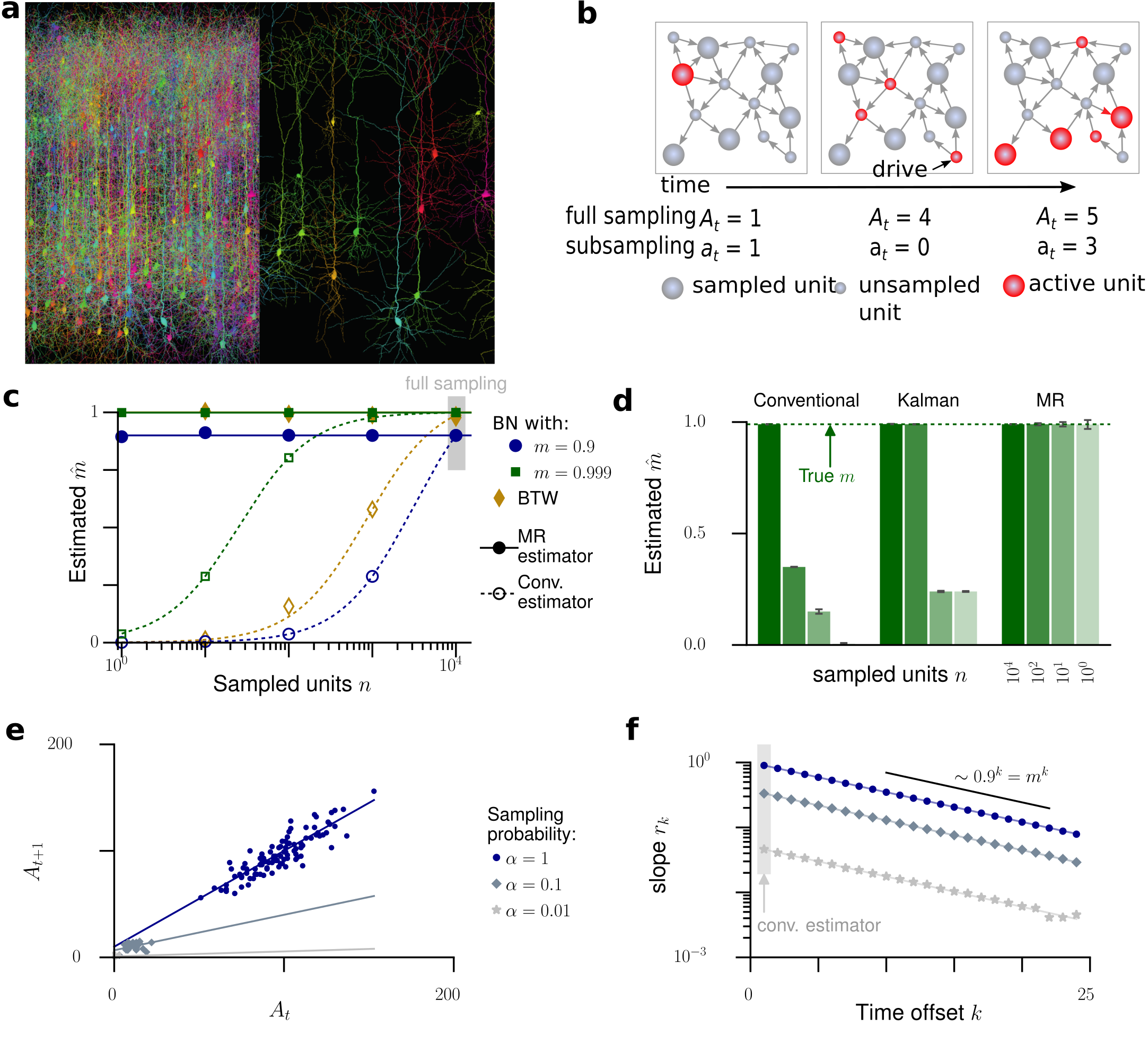}
\caption{\textbf{Spatial subsampling.}
\textbf{a}. In complex networks, such as the brain, often only a small subset of all units can be sampled (spatial subsampling); figure created using TREES \autocite{Cuntz2010}. 
\textbf{b}. In a branching network (BN), an active unit (e.g. a spiking neuron, infected individual, or defaulting bank) activates some of its neighbors in the next time step.
Thereby activity can spread over the system.
Units can also be activated by external drive.
As the subsampled activity $a_{t}$ may significantly differ from the actual activity $A_{t}$, spatial subsampling can impair inferences about the dynamical properties of the full system.
\textbf{c}. In recurrent networks (BN, Bak-Tang-Wiesenfeld model (BTW)), the conventional estimator (empty symbols) substantially underestimates the branching ratio $m$ when less units $n$ are sampled, as theoretically predicted (dashed lines). The novel multistep regression (MR) estimator (full symbols) always returns the correct estimate, even when sampling only 10 or 1 out of all $N=10^4$ units.
\textbf{d}. For a BN with $m=0.99$, the conventional estimator infers $\mh = 0.37$, $\mh = 0.1$ or $\mh = 0.02$ when sampling 100, 10, or 1 units respectively. Kalman filtering based estimation returns approximately correct values under slight subsampling ($n=100$), but is biased under strong subsampling. In contrast, MR estimation returns the correct $\mh$ for any subsampling. 
\textbf{e}. MR estimation is exemplified for a subcritical branching process ($m=0.9$, $h=10$), where active units are observed with probability $\alpha$.
Under subsampling (gray), the regression slopes $r_1$ are smaller than under full sampling (blue).
\textbf{f}. While conventional estimation of $m$ relies on the linear regression $r_1$ and is biased under subsampling, MR estimation infers $\hat{m}$ from the exponential relation $r_k \propto m^k$, which remains invariant under subsampling.
}
\label{fig:branching_cartoon}
\end{figure*}

\noindent

A conventional estimator \autocite{Heyde1972,Wei1990} $\mh_\mathrm{C}$ of $m$ uses linear regression of activity at time $t$ and $t+1$, because the slope of linear regression directly returns $m$ owing to the autoregressive representation in Eq. \eqref{eq:autoreg}.
This estimation of $m$ is consistent if the full activity $A_{t}$ is known.
However, under subsampling it can be strongly biased, as we show here.
To derive the bias quantitatively, we model subsampling in a generic manner in our stochastic framework: We assume only that the subsampled activity $a_{t}$ is a random variable that \emph{in expectation} it is proportional to $A_t$, $\E{a_{t} \cond A_{t}} = \alpha \, A_{t} + \beta$ with two constants $\alpha$ and $\beta$ (\ref{sec:supp_subsampling}).
This represents, for example, sampling a fraction $\alpha$ of all neurons in a brain area.
Then the conventional estimator is biased by $m \, ( \alpha^2 \Var{A_{t}} \, / \, \Var{a_{t}} - 1 )$ (Corollary \ref{theorem:regression_biased}).
The bias vanishes only when all units are sampled ($\alpha = 1$, Figs. \ref{fig:branching_cartoon}\textbf{c}--\textbf{e}), but is inherent to subsampling and cannot be overcome by obtaining longer recordings.

Kalman filtering \autocite{Hamilton1994,Shumway1982,Ghahramani1996}, a state-of-the-art approach for system identification, cannot overcome the subsampling bias either, because it assumes Gaussian noise for both the evolution of $A_t$ and the sampling process for generating $a_t$ (see \ref{sec:supp_kalman}).
These assumptions are violated under typical subsampling conditions, when the values of $a_t$ become too small, so that the central limit theorem is not applicable, and hence Kalman filtering fails  (Figs. \ref{fig:branching_cartoon}\textbf{d}, \ref{fig:supp_kalman}).
It is thus applicable to a much narrower set of subsampling problems and in addition requires orders of magnitude longer runtime compared to our novel estimator (Fig. \ref{fig:supp_kalman}).
 
Our novel estimator takes a different approach than the other estimators (\ref{sec:supp_mr}).
Instead of directly using the biased regression of activity at time $t$ and $t+1$, we perform multiple linear regressions of activity between times $t$ and $t+k$ with different time lags $k=1,\ldots,k_\mathrm{max}$.
These return a collection of linear regression slopes $r_k$ (note that $r_1$ is simply the conventional estimator $\hat{m}_\mathrm{C}$).
Under full sampling, one expects an exponential relation \autocite{Statman2014} $r_k = m^k$ (Theorem \ref{lemma:mlr_slopes}).
Under subsampling, however, we showed that all regressions slopes $r_k$ between $a_t$ and $a_{t+k}$ are biased \emph{by the same factor} $b = \alpha^2 \Var{A_{t}} \, / \, \Var{a_{t}}$ (Theorem \ref{theorem:mlr_under_subsampling}).
Hence, the exponential relation generalizes to

\begin{equation}
r_k = \alpha^2 \frac{\Var{A_t}}{\Var{a_t}} \, m^k = b \, m^k
\label{eq:mr_main}
\end{equation}

\noindent
under subsampling.
The factor $b$ is, in general, not known and thus $m$ cannot be estimated from any $r_k$ alone.
However, because $b$ is constant, one does not need to know $b$ to estimate $\mh$ from regressing the collection of slopes $r_k$ against the exponential model $b \, m^k$ according to Eq. \eqref{eq:mr_main}. This result serves as the heart of our new multiple regression (MR) estimator (Figs. \ref{fig:branching_cartoon}\textbf{f}, \ref{fig:supp_ar_processes}, \ref{fig:supp_consistency_transients}, Corollary \ref{theorem:mlr_unbiased} and Theorem \ref{theorem:mlr_under_subsampling}).

 \begin{figure*}
 \centering
\includegraphics[width=120mm]{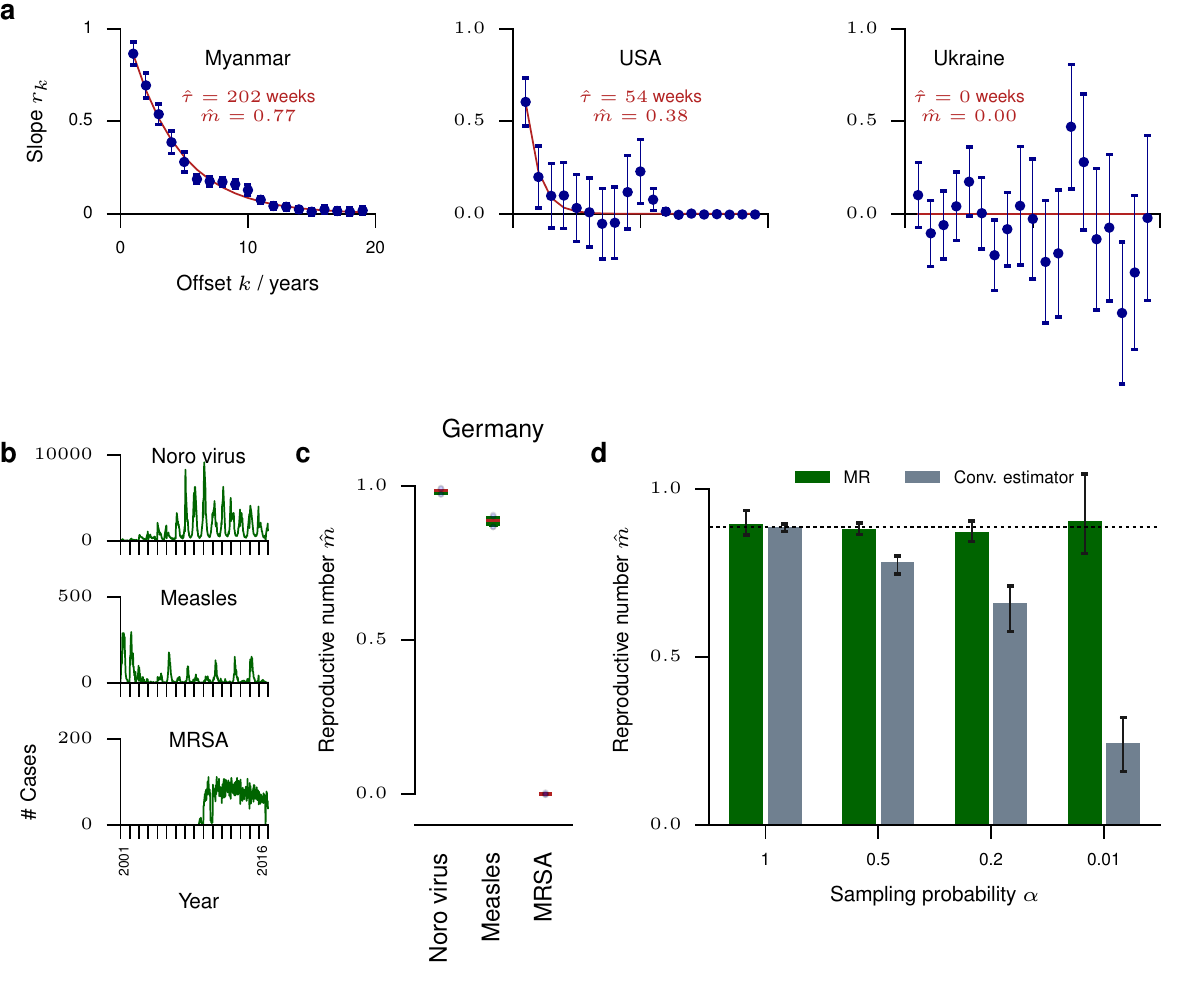}
\caption{\textbf{Disease propagation.}
In epidemic models, the reproductive number $m$ can serve as an indicator for the infectiousness of a disease within a population, and predict the risk of large incidence bursts.
We have estimated $\mh$ from incidence time series of measles infections for 124 countries worldwide (\ref{sec:supp_models}); as well as noroviral infection, measles, and invasive meticillin-resistant Staphylococcus aureus (MRSA) infections in Germany.
\textbf{a.} MR estimation of $\hat{m}$ is shown for measles infections in three different countries. Error bars here and in all following figures indicate 1SD or the corresponding 16\% to 84\% confidence intervals if asymmetric. The reproductive numbers $\hat{m}$ decrease with the vaccination rate (Spearman rank correlation: $r = -0.342, p < 10^{-4}$).
\textbf{b.} Weekly case report time series for norovirus, measles and MRSA in Germany. 
\textbf{c.} Reproductive numbers $\hat{m}$ for these infections.
\textbf{d.} When artificially subsampling the measles recording (under-ascertainment), conventional estimation underestimates $\mh_\mathrm{C}$, while MR estimation still returns the correct value. Both estimator return the same $\mh$ under full sampling.
}
\label{fig:diseases}
\end{figure*}

In fact, MR estimation is equivalent to estimating the autocorrelation time of subcritical PARs, where autocorrelation and regression $r_k$ are equal:  We showed that subsampling decreases the autocorrelation \textit{strength} $r_k$, but the autocorrelation \textit{time} $\tau$ is preserved.
This is because the system itself evolves independently of the sampling process.
While subsampling biases each regression $r_k$ by decreasing the mutual dependence between subsequent observations ($a_{t}, a_{t+k}$), the temporal decay in $r_k \sim m^k = e^{-k \, \Delta t \, / \, \tau}$ remains unaffected, allowing for a consistent estimate of $m$ even when sampling only a single unit (Fig. \ref{fig:branching_cartoon}\textbf{d}).
Particularly close to $m=1$ the autocorrelation time $\tau = - \Delta t \, / \, \log m $ diverges, which is known as critical slowing down \autocite{Scheffer2012}. Because of this divergence, MR estimation can resolve the distance to criticality in this regime with high precision.

The MR estimator is consistent under subsampling, because the system itself evolves independently of the sampling process:
While subsampling biases each regression $r_k$ by decreasing the mutual dependence between subsequent observations ($a_{t}, a_{t+k}$), the temporal decay in $r_k \sim m^k = e^{-k \, \Delta t \, / \, \tau}$ remains unaffected. Here, $\tau = - \Delta t \, / \, \log m $ refers to the autocorrelation time of stationary (subcritical) processes, where autocorrelation and regression $r_k$ are equal, and $\Delta t$ is the time scale of the investigated process. Thus for subcritical PARs, subsampling decreases the autocorrelation \textit{strength} $r_k$, while the autocorrelation \textit{time} $\tau$ is preserved.    
Making use of this result allows for a consistent estimate of $m$ even when sampling only a single unit (Fig. \ref{fig:branching_cartoon}\textbf{d}).


PARs are typically only a first order approximation of real world event propagation.
However, their mathematical structure allowed for an analytical derivation of the subsampling bias and the consistent estimator.
To show that the MR estimator returns correct results also for more complex systems, we applied it to more complex simulated systems:
a branching network \autocite{Haldeman2005} (BN) and the non-linear Bak-Tang-Wiesenfeld model \autocite{Bak1987} (BTW).
In contrast to generic PARs, these models (a) run on recurrent networks and (b) are of finite size. 
In addition, the second model shows (c) completely deterministic propagation of activity instead of the stochastic propagation that characterizes PARs, and (d) the activity of each unit depends on many past time steps, not only one.
Both models approximate neural activity propagation in cortex \autocite{Beggs2003, Haldeman2005,Priesemann2009, Ribeiro2010, Priesemann2013, Priesemann2014}.
For both models the numerical estimates of $m$ were precisely biased as analytically predicted, although the models are only approximated by a PAR (dashed lines in Fig. \ref{fig:branching_cartoon}\textbf{c}, Eq. \ref{eq:hypergeometric_bias}). 
The bias is considerable: For example, sampling 10\% or 1\% of the neurons in a BN with $m=0.9$ resulted in the estimates $\sh_\mathrm{C}=r_1 = 0.312$,  or even $\sh_\mathrm{C}=0.047$, respectively.
Thus a process fairly close to instability ($m=0.9$) is mistaken as Poisson-like ($\mh_\mathrm{C}=0.047 \approx 0$) just because sampling is constrained to 1\% of the units.
Thereby the risk that systems may develop instabilities is severely underestimated.

MR estimation is readily applicable to subsampled data, because it only requires a sufficiently long time series $a_{t}$, and the assumption that in expectation $a_{t}$ is proportional to $A_{t}$.
Hence, in general it suffices to sample the system randomly, without even knowing the system size $N$, the number of sampled units $n$, or any moments of the underlying process.
Importantly, one can obtain a consistent estimate of $m$, even when sampling only a very small fraction of the system, under homogeneity even when sampling only one single unit (Figs. \ref{fig:branching_cartoon}\textbf{c},\textbf{d}, Fig. \ref{fig:supp_cat_single_electrodes}). This robustness makes the estimator readily applicable to any system that can be approximated by a PAR.
We demonstrate the bias of conventional estimation and the robustness of MR estimation at the example of two real-world applications.


\paragraph{Application to disease case reports.}


We used the MR estimator to infer the ``reproductive number'' $m$ from incidence time series of different diseases \autocite{Diekmann1990}.
Disease propagation represents a nonlinear, complex, real-world system often approximated by a PAR \autocite{Earn2000,Brockmann2006}.
Here, $m$ determines the disease spreading behavior and has been deployed to predict the risk of epidemic outbreaks \autocite{Farrington2003}.
However, the problem of subsampling or \textit{under}-\linebreak \textit{ascertainment} has always posed a challenge \autocite{Papoz1996, Hauri2011}.

As a first step, we cross-validated the novel against the conventional estimator using the spread of measles in Germany, surveyed by the Robert-Koch-Institute (RKI). We chose this reference case, because we expected case reports to be almost fully sampled owing to the strict reporting policy supported by child care facilities and schools \autocite{Hellenbrand2003, Wichmann2009}, and to the clarity of symptoms. 
Indeed, the values for $\mh$ inferred with the conventional and with the novel estimator, coincided (Fig. \ref{fig:diseases}\textbf{d}, \ref{sec:supp_epidemiology}). 
In contrast, after applying artificial subsampling to the case reports, thereby mimicking that each infection was only diagnosed and reported with probability $\alpha < 1$, the conventional estimator severely underestimated the spreading behavior, while MR estimation always returned consistent values (Fig. \ref{fig:diseases}\textbf{d}).
This shows that the MR estimator correctly infers the reproductive number $m$ directly from subsampled time series, without the need to know the degree of under-ascertainment $\alpha$.

Second, we evaluated worldwide measles case and vaccination reports for 124 countries provided by the WHO since 1980 (Fig. \ref{fig:diseases}\textbf{a}, \ref{sec:supp_epidemiology}), because the vaccination percentage differs in each country, and this is expected to impact the spreading behavior through $m$. 
The  reproductive numbers $\mh$ ranged between 0 and 0.93, and in line with our prediction clearly decreased with increasing vaccination percentage in the respective country (Spearman rank correlation: $r = -0.342, p < 10^{-4}$).

Third, we estimated the reproductive numbers for three diseases in Germany with highly different infectiousness: noroviral infection \autocite{Hauri2011,Bernard2014}, measles, and invasive meticillin-resistant Staphylococcus aureus (MRSA, an antibiotic-resistant germ classically associated with health care facilities \autocite{Boucher2008}, Figs. \ref{fig:diseases}\textbf{b},\textbf{c}),
and quantified their propagation behavior.
MR estimation returned the highest $\hat{m} = 0.98$ for norovirus, compliant with its high infectiousness \autocite{Teunis2008}.
For measles we found the intermediate $\hat{m} = 0.88$, reflecting the vaccination rate of about 97\%.
For MRSA we identified $m=0$, confirming that transmission is still minor in Germany \autocite{Kock2011}.
However, a future increase of transmission is feared and would pose a major public health risk \autocite{DeLeo2010}.
Such an increase could be detected by our estimator, even in countries where case reports are incomplete.

\paragraph{Reverberating spiking activity \textit{in vivo}}

\begin{figure*}
 \centering
\includegraphics[width=\textwidth]{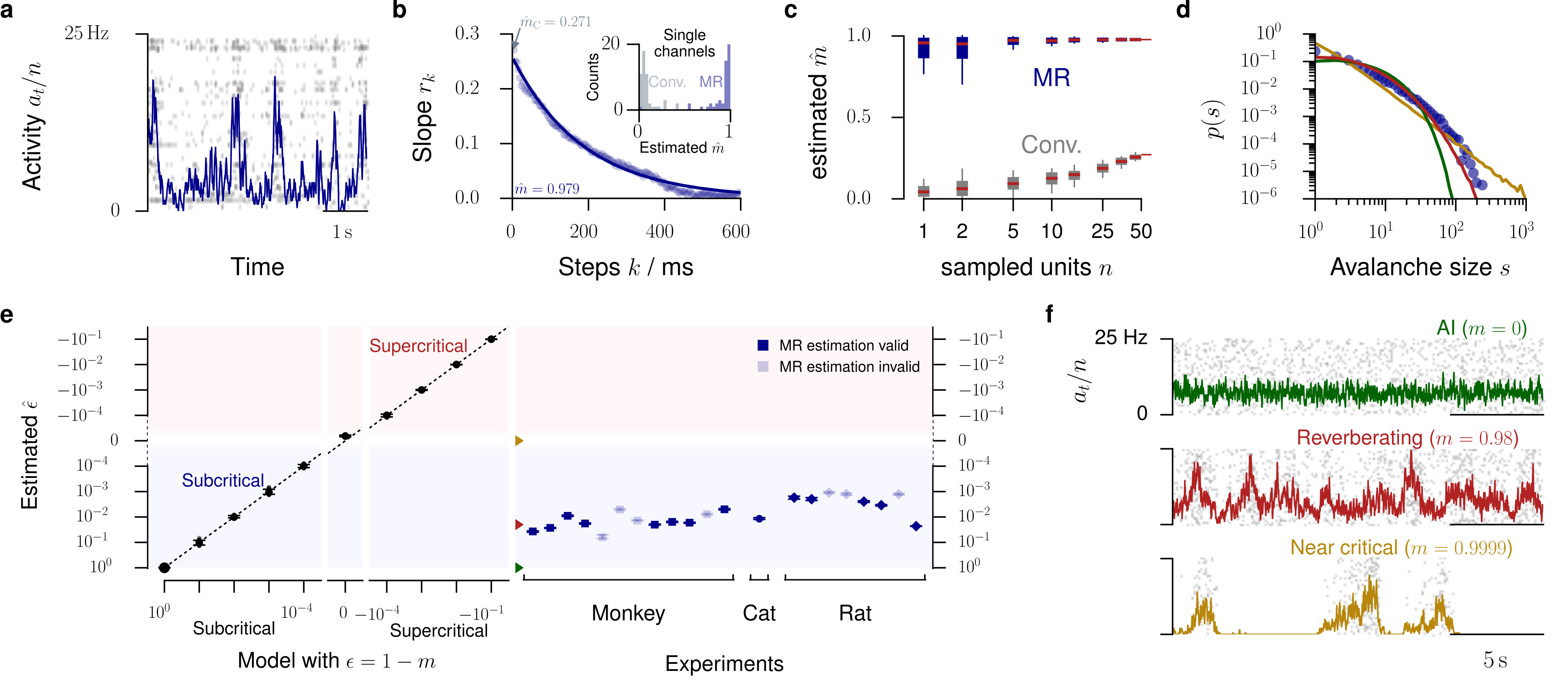}
 \caption{\textbf{Animal spiking activity \textit{in vivo}.} 
In neuroscience, $m$ denotes the mean number of spikes triggered by one spike.
We estimated $\mh$ from spiking activity recorded \textit{in vivo} in monkey prefrontal cortex, cat visual cortex, and rat hippocampus.
\textbf{a}. Raster spike plot and population rate $a_{t}$ of 50 single units illustrated for cat visual cortex.
\textbf{b}. MR estimation based on the exponential decay of the autocorrelation of $r_k$ of $a_{t}$. Inset: Comparison of conventional and MR estimation results for single units (medians $\mh_\mathrm{C} = 0.057$ and $\mh = 0.954$ respectively).
\textbf{c} $\mh$ estimated from from further subsampled cat recordings, estimated with the conventional and MR estimator. Error bars indicate variability over 50 randomly subsampled $n$ out of the recorded 50 channels.
\textbf{d} Avalanche size distributions for cat visual cortex (blue) and the networks with AI, reverberating and near-critical dynamics in panel \textbf{f}.
\textbf{e}. For all simulations, MR estimation returned the correct distance to instability (criticality) $\epsilon = 1 - m$ (\ref{sec:supp_models}). \textit{In vivo} spike recordings from rat, cat, and monkey, clearly differed from critical ($\epsilon = 0$) and AI ($\epsilon=1$) states (median $\sh = 0.98$, error bars: 16\% to 84\% confidence intervals, note that some confidence intervals are too small to be resolved). Opaque symbols indicate that MR estimation was rejected (Fig. \ref{fig:supp_animal_data}, \ref{sec:supp_poisson}). Green, red, and yellow arrows indicate $\epsilon$ for the dynamic states shown in panel \textbf{f}.
\textbf{f}. Population activity and raster plots for AI activity, reverberating, \textit{in vivo}-like, and near critical networks. All three networks match the recording from cat visual cortex with respect to number of recorded neurons and mean firing rate.
}
\label{fig:animals}
\end{figure*}

We applied the MR estimator to cortical spiking activity \textit{in vivo} to investigate two contradictory hypothesis about the collective spiking dynamics.
One hypothesis suggests that the collective dynamics is ``asynchronous irregular'' (AI) \autocite{Burns1976,Softky1993,Steveninck1997,Ecker2010}, i.e. neurons spike independently of each other and in a Poisson manner ($m=0$), which may reflect a balanced state \autocite{Vreeswijk1996a,Brunel2000,Renart2010}.
The other hypothesis suggests that neuronal networks operate at criticality ($m=1$) \autocite{Priesemann2009,Beggs2003,Levina2007,Chialvo2010,Tkacik2014,Humplik2017}, thus in a particularly sensitive state close to a phase transition. These different hypotheses have distinct implications for the coding strategy of the brain:
Criticality is characterized by long-range correlations in space and time, and in models optimizes performance in tasks that profit from  long reverberation of the activity in the network \autocite{Haldeman2005,Kinouchi2006,Boedecker2012,Shew2013,DelPapa2017}.
In contrast, the typical balanced state minimizes redundancy \autocite{Hyvarinen2000} and supports fast network responses \autocite{Vreeswijk1996a}.

Analyzing \textit{in vivo} spiking activity from Macaque monkey prefrontal cortex during a memory task, anesthetized cat visual cortex with no stimulus (Figs. \ref{fig:animals}\textbf{a},\textbf{b}), and rat hippocampus during a foraging task (\ref{sec:supp_animals}) returned $\sh$ to be between $0.963$ and $0.998$ (median $\sh=0.984$, Figs. \ref{fig:animals}\textbf{e}, \ref{fig:supp_animal_data}), corresponding to autocorrelation times between \SI{100}{ms} and \SI{2000}{ms}. 
This clearly suggests that spiking activity \textit{in vivo} is neither AI-like ($m=0$), nor consistent with a critical state ($m=1$), but in a reverberating state that shows autocorrelation times of a few hundred milliseconds.
We call the range of the dynamical states found in vivo \textit{reverberating}, because input reverberates for a few hundred millisecond in the network, and therefore enables integration of information \autocite{Murray2014a,Chaudhuri2015,Jaeger2004}. Thereby the reverberating state constitutes a specific narrow window between AI state, where perturbations of the firing rate are quenched immediately,  and the critical state, in which perturbations can in principle persist infinitely long (for more details, see \textcite{Wilting2018}).

We demonstrate the robustness to subsampling for the activity in cat visual cortex: we chose random subsets of $n$ neurons from the total of 50 recorded single units.
For any subset, even for single neurons, MR estimation returned about the same median $\mh$ (Fig. \ref{fig:animals}\textbf{c}).
In contrast, the conventional estimator misclassified neuronal activity by strongly underestimating $\mh$: instead of $\mh = 0.984$, it returned $\mh_\mathrm{C} = 0.271$ for the activity of all 50 neurons.
This underestimation gets even more severe when considering stronger subsampling ($n < 50$, Fig. \ref{fig:animals}\textbf{c}).
Ultimately, for single neuron activity, the conventional estimator returned $\hat{m}_\mathrm{C} = 0.057 \approx 0$, which would spuriously indicate dynamics close to AI instead of the reverberating state (inset of Fig. \ref{fig:animals}\textbf{b}, Figs. \ref{fig:animals}\textbf{c} and \ref{fig:supp_cat_single_electrodes}).
The underestimation of $\mh_\mathrm{C}$ was present in all experimental recordings ($r_1$ in Fig. \ref{fig:supp_animal_data}).

On first sight, $\sh=0.984$ may appear close to the critical state, particularly as physiologically a 1.6\% difference to $m=1$ is small in terms of the effective synaptic strength. 
However, this seemingly small difference in single unit properties has a large impact on the \emph{collective} dynamics and makes AI, reverberating, and critical states clearly distinct.
This distinction is readily manifest in the fluctuations of the population activity (Fig. \ref{fig:animals}\textbf{f}).
Furthermore, the distributions of avalanche sizes clearly differ from the power-law  scaling expected for critical systems \autocite{Beggs2003}, but are well captured by a matched, reverberating model (Fig. \ref{fig:animals}\textbf{d}).
Because of the large difference in the network dynamics, the MR estimator can distinguish AI, reverberating, and critical states with the necessary precision.
In fact, the estimator would allow for 100 times higher precision when distinguishing critical from non-critical states, assuming \textit{in vivo}-like subsampling and mean firing rate (sampling $n=100$ from $N = 10^4$  neurons, Fig. \ref{fig:animals}\textbf{e}).
With larger $N$, this discrimination becomes even more sensitive (detailed error estimates: Fig. \ref{fig:supp_variance_estimation} and \ref{sec:supp_variance}).
As the number of neurons in a given brain area is typically much higher than $N=10^4$ in the simulation, finite size effects are not likely to account for the observed deviation from criticality $\epsilon = 1-m \approx 10^{-2}$ \textit{in vivo}, supporting that in rat, cat, and monkey the brain does \emph{not} operate in a critical state.
Still, additional factors like input or refractory periods may limit the maximum attainable $m$ to quasi-critical dynamics on a Widom line \cite{Williams-Garcia2014}, which could in principle conform with our results.

Most real-world systems, including disease propagation or cortical dynamics, are more complicated than a simple PAR.
For cortical dynamics, for example, heterogeneity of neuronal morphology and function, non-trivial network topology, and the complexity of neurons themselves are likely to have a profound impact onto the population dynamics \autocite{Marom2010}.
In order to test for the applicability of a PAR approximation, we defined a set of conservative tests (\ref{sec:supp_poisson}) and included only those time series, where the approximation by a PAR was considered appropriate.
For example, we excluded all recordings that showed an offset in the slopes $r_k$, because this offset is, strictly speaking, not explained by a PAR and might indicate non-stationarities (Fig. \ref{fig:supp_monkey_nonstationaritites}).
Even with these conservative tests, we found the exponential relation $r_k = b \, m^k$ expected for PARs in the majority of real-world time series
(Fig. \ref{fig:supp_animal_data}, \ref{sec:supp_epidemiology}). 
This shows that a PAR is a reasonable approximation for dynamics as complex as cortical activity or disease propagation. With using PARs, we draw on the powerful advantage of analytical tractability, which allowed for valuable insight into dynamics and stability of the respective system. It is then a logical next step to refine the model by including additional relevant parameters \autocite{Eckmann2007}.
However, the increasing richness of detail typically comes at the expense of analytical tractability.

By employing for the first time a consistent, quantitative estimation, we provided evidence that \textit{in vivo} spiking population dynamics reflects a stable, fading reverberation state around $m=0.98$ universally across different species, brain areas, and cognitive states.
Because of its broad applicability, we expect that besides the questions investigated here, MR estimation can substantially contribute to the understanding of real-world dynamical systems in diverse fields of research where subsampling prevails.

\beginsupplement

\onecolumn

\setcounter{page}{1}

\section*{Supplementary material}

\subsection{Applicability of MR estimation} 
\label{sec:supp_applicability}
We here analytically derive the novel MR estimator for branching processes (BP) \autocite{Harris1963,Heathcote1965,Pakes1971}. We expect that analogous derivations apply to any process with a first order autoregressive representation (PAR) \autocite{Ispany2010}, because these processes fulfill Eq. \eqref{eq:PAR_equation}.
Beside BPs, PARs  include  autoregressive AR(1) processes, integer-valued autoregressive INAR(1) processes \autocite{Alzaid1990} rounded integer-valued autoregressive RINAR(1) processes \autocite{Kachour2009}, and Kesten processes \autocite{Kesten1973a}.


We emphasize that the MR estimator only requires the subsampled recording $a_{t}$ of a system with full activity $A_{t}$ conforming with the definition below.
It is not necessary to know either the full system size, the number of subsampled units, nor any of the moments of the full process $A_{t}$.

\subsection{Branching processes}
\label{sec:supp_bps}
In a branching process (BP) with immigration \autocite{Harris1963,Heathcote1965,Pakes1971} each unit $i$ produces a random number $y_{t,i}$ of units in the subsequent time step.
Additionally, in each time step a random number $h_t$ of units immigrates into the system (drive).  
Mathematically, BPs are defined as follows \autocite{Harris1963, Heathcote1965}: 
Let $y_{t,i}$ be independently and identically distributed non-negative integer-valued random variables following a law $\mathcal{Y}$ with mean $m = \E{\mathcal{Y}}$ and variance $\sigma^2 = \Var{\mathcal{Y}}$.
Further, $\mathcal{Y}$ shall be non-trivial, meaning it satisfies $\Prob{\mathcal{Y} = 0} > 0$ and $\Prob{\mathcal{Y} = 0} + \Prob{\mathcal{Y} = 1} < 1$.
Likewise, let $h_t$ be independently and identically distributed non-negative integer-valued random variables following a law $\mathcal{H}$ with mean rate $h = \E{\mathcal{H}}$ and variance $\xi^2 = \Var{\mathcal{H}}$.
Then the evolution of the BP $A_{t}$ is given recursively by

\begin{equation}
A_{t+1} = \sum_{i = 1}^{A_{t}} y_{t,i} + h_t,
\label{eq:supp_branching_process}
\end{equation}
i.e. the number of units in the next generation is given by the offspring of all present units and those that were introduced to the system from outside.

The stability of BPs is solely governed by the mean offspring $m$. In the subcritical state, $m<1$, the population converges to a stationary distribution $A_\infty$ with mean $\E{A_\infty} = h / (1 - m)$.
At criticality ($m=1$), $A_{t}$ asymptotically exhibits linear growth, while in the supercritical state ($m>1$) it grows exponentially.
We will first show results that further specify the mean and variance of subcritical branching processes.

\begin{lemma}
\label{theorem:exp_and_var}
The stationary distribution of a subcritical BP satisfies

\begin{align*}
\E{A_\infty} = \frac{h}{1-m}, \qquad \Var{A_\infty} = \frac{1}{1-m^2} \left( \xi^2 + \sigma^2 \frac{h}{1-m} \right),
\end{align*}

\noindent
where $m$, $\sigma^2$, $h$, and $\xi^2$ are defined as above. 
\end{lemma}

\begin{proof}
The first result was stated before \autocite{Heathcote1965, Heyde1972} and follows from taking expectation values of both sides of Eq. \eqref{eq:supp_branching_process}: $\E{A_{t+1}} = m \E{A_{t}} + h$. Because of stationarity $\E{A_{t+1}} = \E{A_{t}} = \E{A_\infty}$ and the result follows easily.
For the second result, observe that by the theorem of total variance, $\Var{A_{t+1}} = \E{\Var{A_{t+1} \cond A_{t}}} + \Var{\E{A_{t+1} \cond A_{t}}}$, where $\E{\cdot}$ denotes the expected value, and $A_{t+1} \cond A_{t}$ conditioning the random variable $A_{t + 1}$ on $A_{t}$.
Because $A_{t+1}$ is the sum of independent random variables, the variances also sum: $ \Var{A_{t+1} \cond A_{t}} = \sigma ^2 \, A_{t} + \xi^2$.
Using the result for $\E{A_\infty}$ one then obtains

\begin{align}
\Var{A_{t+1}} = \xi^2 + \sigma ^2 \frac{h}{1-m} + \Var{ m A_{t} + h} =  \xi^2 + \sigma ^2 \frac{h}{1-m} + m^2 \Var{A_{t}}.
\end{align}

\noindent
Again, in the stationary distribution $\Var{A_{t+1}} = \Var{A_{t}} = \Var{A_\infty}$ and hence the stated result follows.
\end{proof}

\subsection{Subsampling}
\label{sec:supp_subsampling}
To derive the MR estimator for subsampled data, subsampling is implemented in a parsimonious way, according to the following definition: 

\begin{definition}[Subsampling]
\label{def:subsampling}
Let $\lbrace A_{t}\rbrace_{t\in \BbbN}$ be a BP and $\lbrace a_{t}\rbrace_{t\in \BbbN}$ a sequence of random variables.
Then $\lbrace a_{t}\rbrace_{t\in \BbbN}$ is called a subsampling of $\lbrace A_{t}\rbrace_{t\in \BbbN}$ if it fulfills the following three conditions:
\begin{enumerate}[(i)]
\item Let $t^\prime, t \in \BbbN$, $t^\prime \neq t$. Then the conditional random variables\footnotemark[2] $(a_{t} | A_{t} = j)$ and $(a_{t^\prime} | A_{t^\prime} = l)$ are independent for any outcome $j, l \in \BbbN$ of $A_t, A_{t'}$. 
 If $A_{t} = A_{t^\prime}$ then $(a_{t} | A_{t} = j)$ and $(a_{t^\prime} | A_{t^\prime} = j)$ are identically distributed.
\item Let $t \in \BbbN$. Conditioning on $a_{t}$ does not add further information to the process: The two random variables $(A_{t+1} \cond A_{t} = j, a_{t} = l)$ and  $(A_{t+1} \cond A_{t}=j)$ are identically distributed for any $j,l \in \BbbN$.
\item There are constants $\alpha, \, \beta \in \BbbR$, $\alpha \neq 0$, such that $\E{a_{t} \cond A_{t} = j} = \alpha j + \beta$ for all $t,j \in \BbbN$. 
\end{enumerate}
\end{definition}

\footnotetext[2]{
Throughout this manuscript, the conditional random variable $(a_{t} | A_{t} = j)$ is to be read as ``$a_t$ given the realization $A_t = j$ of the random variable $A_t$''.}

\noindent
Thus the subsample $a_{t}$ is constructed from the full process $A_{t}$ based on the three assumptions: 
(i) The sampling process does not interfere with itself, and does not change over time. Hence the realization of a subsample at one time does not influence the realization of a subsample at another time, and the conditional \textit{distribution} of $(a_{t}|A_{t})$ is the same as $(a_{t^\prime}|A_{t^\prime})$ if $A_{t} = A_{t^\prime}$. However, even if $A_{t} = A_{t^\prime}$, the subsampled $a_{t}$ and $a_{t'}$ do not necessarily take the same value.
(ii) The subsampling does not interfere with the evolution of $A_{t}$, i.e. the process evolves independent of the sampling. (iii) \textit{On average} $a_{t}$ is proportional to $A_{t}$ up to a constant term. 

It will be shown later, that the novel estimator is applicable to any time series $a_{t}$ that was acquired from a BP conforming with this definition of subsampling.
We will demonstrate possible applications at the hand of two examples:

\paragraph{1. Diagnosing infections with probability $\alpha$.}
For example, when a BP $A_{t}$ represents the spread of infections within a population, each infection may be diagnosed with probability $\alpha \leq 1$, depending on the sensitivity of the test and the likelihood that an infected person consults a doctor.
If each of the $A_{t}$ infections is diagnosed independently of the others, then the number of diagnosed cases $a_{t}$ follows a binomial distribution $a_{t} \sim \mathrm{Bin}(A_{t}, \alpha)$.
Then $\E{a_{t} | A_{t} = j} = \alpha \, j$ is given by the expected value of the binomial distribution.
This implementation of subsampling conforms with the definition above, with the sampling probability $\alpha$ and the constant in (iii) being identical here.

\paragraph{2. Sampling a subset of system components.} In a different application, assume a high-dimensional system of interacting units that forms the substrate on which activation propagates.
Often, the states of a subset of units are observed continuously, for example by placing electrodes that record the activity of the same set of neurons over the entire recording (Fig. \ref{fig:branching_cartoon}\textbf{b}).
This implementation of subsampling in finite size systems is mathematically approximated as follows: 
If $n$ out of all $N$ model units are sampled, the probability to sample $a_{t}$ active units out of the actual $A_{t}$ active units follows a hypergeometric distribution, $a_{t} \sim \mathrm{Hyp}(N, n, A_{t})$.
As $\E{a_{t} \cond A_{t} = j} = j \, n \, / \, N$, this representation satisfies Def. \ref{def:subsampling} with $\alpha = n \, / \, N$.
Choosing this special implementation of subsampling allows to evaluate $\Var{a_{t}}$ further in terms of $A_{t}$:

\begin{align}
\Var{a_{t}} = & \: \E{\Var{a_{t} \cond A_{t}}} + \Var{\E{a_{t} \cond A_{t}}} \nonumber \\
= & \: n \E{ \frac{A_{t}}{N}\frac{N - A_{t}}{N} \frac{N-n}{N-1} } + \Var{\frac{n}{N}A_{t}} \nonumber \\
= & \: \frac{1}{N}\frac{n}{N}\frac{N-n}{N-1} \left( N \, \E{A_{t}} - \E{A_{t}^2} \right) + \frac{n^2}{N^2} \Var{A_{t}} \nonumber \\
= & \: \frac{n}{N^2}\frac{N-n}{N-1} \left( N \, \E{A_{t}} - \E{A_{t}}^2 \right) + \left( \frac{n^2}{N^2} -  \frac{n}{N^2}\frac{N-n}{N-1} \right) \Var{A_{t}}.
\label{eq:subsampled_variance}
\end{align}

\noindent
This expression precisely determines the variance $\Var{a_{t}}$ under subsampling from the properties $\E{A_{t}}$ and $\Var{A_{t}}$ of the full process (which for BPs are known from Lemma \ref{theorem:exp_and_var}), and from the parameters of subsampling $n$ and $N$.
Using Eq. \eqref{eq:subsampled_variance}, we could predict the linear regression slopes $\hat{r}_k$ under subsampling (Theorem \ref{theorem:mlr_under_subsampling}, Eq. \eqref{eq:slopes_final}) in more detail:


\begin{equation}
r_k = \alpha^2 \frac{\Var{A_{t}}}{\Var{a_{t}}} m^k =  \frac{n (N-1) \Var{A_{t}}}{(N-n) ( N\E{A_{t}} - \E{A_{t}}^2 ) + (nN - N )  \Var{A_{t}}} m^k =: b(N, n, \E{A_{t}}, \Var{A_{t}}) \, m^k.
\label{eq:hypergeometric_bias}
\end{equation}

\noindent
The term $b = b(N, n, \E{A_{t}}, \Var{A_{t}})$ is constant when subsampling a given (stationary) system, and quantifies the factor by which $\sh_\mathrm{C}$ is biased when using the conventional estimate for $m$. It depends on $N$, $n$ and the first two moments of $A_{t}$ and is thus known for a BP. This relation was used for Fig. \ref{fig:branching_cartoon}\textbf{c}.


\subsection{MR estimation}
\label{sec:supp_mr}

We here derive an estimator for the mean offspring $m$ based on the autoregressive representation of the BP,

\begin{equation}
\E{A_{t+1} \cond A_{t} = j} = m \, j + h.
\label{eq:PAR_equation}
\end{equation}

\noindent
This novel estimator is based on multistep regressions \autocite{Statman2014} (MR estimator), which generalize \eqref{eq:PAR_equation} to arbitrary time steps $k$. From iteration of Eq. \eqref{eq:PAR_equation}, it is easy to see that

\begin{equation}
\E{A_{t+k} \cond A_{t} = j} = m^k \, j + h \frac{1-m^k}{1-m}.
\label{eq:PAR_equation_2}
\end{equation}

\begin{definition}[Multistep regression estimator]
\label{def:mr_estimator}
Consider a subsampled BP $\lbrace a_{t} \rbrace$ of length $T$. Let $k_\mathrm{max} \in \BbbN$, $k_\mathrm{max} \geq 2$.
Then multistep regression (of $k_\mathrm{max} $-th order) estimates $m$ in the following way:
\begin{enumerate}
\item For $k=1, \ldots, k_\mathrm{max} $, estimate the slope $\hat{r}_k$ and offset $\hat{s}_k$ of linear regression between the pairs $\lbrace (a_{t},\, a_{t+k}) \rbrace_{t=0}^{T-k}$, e.g. by least square estimation (Fig. \ref{fig:branching_cartoon}\textbf{e}), i.e. by minimizing the residuals
\begin{equation}
R_k(\hat{r}_k, \hat{s}_k) = \sum_{t} \left( a_{t + k} -  (\hat{r}_k \cdot a_t + \hat{s}_k ) \right)^2.
\label{eq:slopes_residual}
\end{equation}

\item Based on the relation \autocite{Statman2014} $r_k = b \cdot m^k$, estimate $\hat{b}$ and $\sh$ by minimizing the sum of residuals 

\begin{equation}
R(\hat{b}, \sh) = \sum_{k=1}^{k_\mathrm{max}} \left( \hat{r}_k -  \hat{b} \cdot \sh^k \right)^2,
\label{eq:mr_residual}
\end{equation}

with the collection of slopes  $\lbrace \hat{r}_k \rbrace_{k=1}^{k_\mathrm{max}}$ obtained from step 1 (Fig. \ref{fig:branching_cartoon}\textbf{f}).

\end{enumerate}
Then $\sh$ is the multistep regression (MR) estimate of the mean offspring $m$. For the application to experimental data, we further applied tests to identify nonstationarities (\ref{sec:supp_poisson}).
\end{definition}

We first prove that the MR estimator is consistent in the fully sampled case, and will then show the consistency under subsampling.
First, we need the following result about the individual linear regression slopes $\hat{r}_k$ under full sampling:

\begin{lemma}
\label{lemma:mlr_slopes}
The slope $\hat{r}_k$, obtained from $A_{t}$ under full sampling, is a consistent estimator for $m^k$. If the process is subcritical, then the offset $\hat{s}_k$ is also a consistent estimator for $h \frac{1-m^k}{1-m}$.
\end{lemma}

\begin{remark}
For $k=1$, these results were already obtained by \cite{Heyde1972,Venkataraman1982a,Wei1990}, and details can be found in these sources. Based on their proofs, we here show the generalization to $k$ timesteps.
\end{remark}

\begin{proof}
Let $k \in \BbbN$, $i \in \lbrace 0, \ldots, k -1 \rbrace$.
Construct a new random process by starting at time $i$ and taking every $k$-th time step of the original process $A_{t}$. This new process is given by $A^{(k, i)}_{t^\prime} = A_{i + k \cdot t^\prime}$ with the index $t^\prime \in \BbbN$.
Hence, the ``time'' $t'$ of this new process relates to the time $t$ of the old process as $t = i + k \cdot t'$.
For a time series of length $T$, let $r^{(k,i)}$  be the least square estimator for the slope and  $\hat{s}^{(k,i)}$ the least square estimator for the intercept of linear regression on all pairs $(A^{(k, i)}_{t^\prime + 1}, \, A^{(k, i)}_{t^\prime})$ from the time series $\lbrace A^{(k, i)}_{t^\prime}\rbrace_{t^\prime=0}^{\lfloor (T - 1)/ k \rfloor}$.
We will derive that $r^{(k,i)}$ is a consistent estimator for $m^k$.
According to \cite{Wei1990}, it is sufficient to show that the evolution of $A^{(k, i)}_{t^\prime}$ can be rewritten as

\begin{equation}
A^{(k, i)}_{t^\prime}  = 
  m^k  \cdot A^{(k, i)}_{t^\prime-1} +  h \frac{1-m^k}{1-m} + \epsilon^{(k, i)}_{t^\prime}
\end{equation}

\noindent
with a martingale difference sequence $\epsilon^{(k, i)}_{t^\prime}$, as this is a stochastic regression equation.
Hence, consider

\begin{equation}
\epsilon^{(k, i)}_{t^\prime} = 
A^{(k, i)}_{t^\prime} - m^k  \cdot A^{(k, i)}_{t^\prime-1} -  h \frac{1-m^k}{1-m}
=
A_{i+k t^\prime} - m^k \cdot A_{i+k \, (t^\prime - 1)} - h \frac{1-m^k}{1-m}.
\end{equation}

\noindent
We now show that $(\epsilon^{(k, i)}_{t^\prime})_{t\prime \in \BbbN}$ is a martingale difference sequence for all $k$.
From iteration of Eq. \eqref{eq:PAR_equation_2}, it is easy to see that

\begin{equation}
\E{A^{(k, i)}_{t^\prime} | A^{(k, i)}_{t^\prime - 1} = j} = \E{A_{k t^\prime + i} | A_{k t^\prime - k + i} = j} = m^k j + h \frac{1-m^k}{1-m}
\label{eq:regression_expectation}
\end{equation}

\noindent
holds.
Hence, $\E{\epsilon^{(k, i)}_{t^\prime} \cond  A^{(k, i)}_{t^\prime-1} =j } =0$
for any $j$ and $\lbrace \epsilon^{(k,i)}_{t^\prime} \rbrace$ is indeed a martingale difference sequence.
Therefore, $\lbrace A^{(k, i)}_{t^\prime}\rbrace_{t^\prime=0}^{\lfloor T / k \rfloor}$ satisfies a linear stochastic regression equation with slope $m^k$ and intercept $h \frac{1-m^k}{1-m}$. 
The least square estimators return unbiased and consistent estimates for the slope and intercept in the subcritical case, i.e. the estimators converge in probability \autocite{Heyde1972,Venkataraman1982a,Wei1990}:

\begin{equation}
\hat{r}^{(k,i)} \convergesinprobability m^k \qquad  \hat{s}^{(k,i)} \convergesinprobability  h \frac{1-m^k}{1-m}. \nonumber
\end{equation}

\noindent
In the critical and supercritical cases, only $\hat{r}^{(k,i)} \convergesinprobability m^k$ holds following \cite{Wei1990}. Hence, we obtain $\hat{r}_k \convergesinprobability m^k$ for all $m$ and $\hat{s}_k \convergesinprobability h (1-m^k)/(1-m)$ if $m < 1$.
\end{proof}
\noindent

\begin{corollary}
\label{theorem:mlr_unbiased}
As least square estimation of $\hat{b}$ and $\hat{m}$ from minimizing the residual \eqref{eq:mr_residual} is consistent, multistep regression is a consistent estimator for $m$ under full sampling, $\hat{m} \convergesinprobability m$.
\end{corollary}

\noindent
These results were obtained for BPs. However, the derivation is here only based on the autoregressive representation \eqref{eq:PAR_equation}, motivation the following proposition:

\begin{proposition}
Multistep regression is a consistent estimator for $m$ for any PAR satisfying Eq. \eqref{eq:PAR_equation}.
\end{proposition}

\noindent
Numerical results for AR(1) and Kesten processes support this conjecture \autocite{Statman2014} (Fig. \ref{fig:supp_ar_processes}).


\noindent
Next, we show that MR estimation is consistent in the subcritical case even if only the subsampled $a_{t}$ is known:

\begin{theorem}
\label{theorem:mlr_under_subsampling}
Let $A_{t}$ be a PAR with $m < 1$ and a stationary limiting distribution $A_\infty$ and let the PAR be started in the stationary distribution, i.e. $A_0 \sim A_\infty$. 
Let $a_{t}$ be a subsampling of $A_{t}$. Multistep regression (MR) on the subsampled $a_{t}$ is a consistent estimator of the mean offspring $m$.
\end{theorem}

\begin{proof}
The existence of a stationary distribution $A_\infty$ was shown by \cite{Heathcote1965}. 
The least square estimator for the slope of linear regression is also given by\autocite{Kenney1962}

\begin{equation}
\hat{r}_k = \hat{\rho}_{a_{t}\,a_{t+k}} \, \frac{\hat{\sigma}_{a_{t}}}{ \hat{\sigma}_{a_{t+k}}}
\end{equation}

\noindent
with the the estimated standard deviations $\hat{\sigma}_{a_{t}}$ and $\hat{\sigma}_{a_{t+k}}$ of $a_{t}$ and $a_{t+k}$ respectively.
In the subcritical state, $\sigma_{a_{t}} = \sigma_{a_{t+k}}$ because of stationarity. Thus estimating the linear regression slope is equivalent to estimating the Pearson correlation coefficient $\hat{\rho}_{a_{t} \, a_{t+k}} = \hat{\rho}_{a_{t}}(k)$ (which is identical to the autocorrelation function of $a_{t}$).
In the following, we calculate the Pearson correlation coefficient for the subsampled time series by evaluating $\E{a_{t} \, a_{t+k}}$. We use the law of total expectation in order to express $\E{a_{t} \, a_{t+k}}$ not in dependence of $a_{t}$, but in terms of $A_{t}$:

\begin{align}
\E{a_{t} \, a_{t+k}} = \: & \E[A_{t+k}, A_{t}]{\E{a_{t}\,a_{t+k} \cond A_{t}, A_{t+k} }},
\end{align}

\noindent
where the inner expectation value is taken with respect to the joint distribution of $a_{t+k}$ and $a_{t}$, and the outer with respect to the joint distribution of $A_{t+k}$ and $A_{t}$.
Through conditioning on both $A_{t}$ and $A_{t+k}$, $(a_{t} \cond A_{t})$ and $(a_{t+k} \cond A_{t+k})$ become independent due to Def. \ref{def:subsampling}. Hence, the joint distribution of $(a_{t}, a_{t+k} \cond A_{t}, A_{t+k})$ factorizes, and the expectation value factorizes as well. 
By definition, $\E{a_{t} \cond A_{t} = j} = \alpha \, j + \beta$ and hence

\begin{align}
\E{a_{t} \, a_{t+k}} =  \E[A_{t+k}, A_{t}]{(\alpha A_{t + k} + \beta)\,(\alpha A_{t} + \beta)}
\end{align}

\noindent

 Without loss of generality, we here show the proof for $\beta=0$ which is easily extended to the general case. We express $\E{a_{t} \,a_{t+k}}$ in terms of Eq. \eqref{eq:PAR_equation_2} using the law of total expectation again: 

\begin{align}
\E{a_{t} \, a_{t+k}} = \: & \alpha^2 \E{A_{t} \, A_{t+k}} \nonumber \\
= \: & \alpha^2 \E[A_{t}]{ \E{A_{t} \, A_{t+k} \cond A_{t} }} \nonumber \\
= \: & \alpha^2 \E[A_{t}]{A_{t} \left( m^k \, A_{t} + h \frac{1 - m^k}{1 - m} \right)} \nonumber \\
 = \: & \alpha^2 \, \left( m^k \, \E{A_{t}^2} + (1-m^k) \, \E{A_{t}}^2 \right), \nonumber
\end{align}

\noindent
where the first expectation was taken with respect to the joint distribution of $A_{t}$ and $A_{t+k}$.
We then used that $\E{A_{t}^2}$ and $\E{A_{t}} = h / (1 - m)$ exist, which follows from stationarity of the process. By a similar argument, 

\begin{align}
\E{a_{t + 1}}  = \E{a_{t}} =  \E[A_{t}]{\E{a_{t} \cond A_{t}}} = \alpha \E{A_{t}} = \alpha \frac{h}{1 - m }
\end{align}

\noindent
and combining these results the covariance is 

\begin{align}
\mathrm{Cov}[a_{t+k}, a_{t}] = \: & \E{a_{t+k} \, a_{t}} - \E{a_{t+k}} \E{a_{t}} =  \alpha^2 \, \left( m^k \, \E{A_{t}^2} + (1-m^k) \, \E{A_{t}}^2 \right) - \alpha^2 \E{A_{t}}^2 = \alpha^2 m^k \Var{A_{t}}. 
\end{align}

\noindent
Therefore, we find that the estimator $\hat{r}_k$ converges in probability: 

\begin{equation}
\hat{r}_k \convergesinprobability \rho_{a_{t}a_{t+k}} = \frac{\mathrm{Cov}[a_{t+k}, a_{t}]}{\Var{a_{t}}} = \alpha^2 \, \frac{\Var{A_{t}}}{\Var{a_{t}}} \, m^k.
\label{eq:slopes_final}
\end{equation}

\noindent
Hence, the bias of of the conventional estimator $\mh_\mathrm{C} = \hat{r}_1$ is precisely given by the factor $b = \alpha^2 \Var{A_{t}} \, / \, \Var{a_{t}}$. However, importantly the relation $\hat{r}_k = \hat{b} \, \hat{m}^k$ still holds for the subsampled $a_{t}$.  Given a collection of multiple linear regressions $\hat{r}_1, \ldots, \hat{r}_{k_{\max}}$, the least square estimation of $\hat{b}$ and $\hat{m}$ from minimizing the residual \eqref{eq:mr_residual} yields a consistent estimator $\hat{m}$ for the mean offspring $m$ even under subsampling and only requires the knowledge of $a_{t}$.
\end{proof}

\noindent
This proof also showed that the conventional estimator \autocite{Heyde1972} is biased under subsampling:
\begin{corollary}
\label{theorem:regression_biased}
Let $\lbrace a_{t} \rbrace$ be a subsampling of a subcritical PAR $\lbrace A_{t} \rbrace$. Then the conventional linear regression estimator $\hat{m}_\mathrm{C} = \hat{r}_1$ by \cite{Heyde1972} is biased by $ m (\alpha^2 \frac{\Var{A_{t}}}{\Var{a_{t}}} - 1)$. Equivalently, it is biased by the factor $\alpha^2 \frac{\Var{A_{t}}}{\Var{a_{t}}}$.
\end{corollary}
\noindent

\paragraph{Nonstationarity, criticality and supercriticality.}
 
The consistency of the estimator in the fully sampled case is included in our proof of Lemma \ref{lemma:mlr_slopes} and follows from the results by \cite{Heyde1972, Wei1990}.
Our proof for the subsampled case (Theorem \ref{theorem:mlr_under_subsampling}), in contrast, strictly requires  stationarity ($A_{t} \sim A_\infty$ for any $t$) and the existence of the first two moments of $A_{t}$.
We expect that the MR estimator is also consistent if the subcritical process is not started in the stationary distribution, $A_0 \nsim A_\infty$, because the results by \cite{Heathcote1965} show that it will converge to this stationary distribution as $t \rightarrow \infty$ (Fig. \ref{fig:supp_consistency_transients}).
Furthermore, numerical results suggest that the MR estimator is also consistent for critical and supercritical cases, where no stationary distribution exists (Fig. \ref{fig:animals}\textbf{d}).

\subsection{Identifying common non-stationarities and Poisson activity.}
\label{sec:supp_poisson}

In many types of analyses, non-stationarities in the time series can lead to wrong results, typically an overestimation of $\mh$. We developed tests to exclude data sets with signatures of common non-stationarities. The different non-stationarities, their impact on the $r_k$ and the rules for rejection of time series are outlined below.

First, \textit{transient} increases of the drive $h_t$, e.g. in response to a stimulus, lead to a transient increase in $\E{A_{t}}$.
These transients induce correlations or anti-correlations, which prevail on long time scales (Fig. \ref{fig:supp_monkey_nonstationaritites}\textbf{c},\textbf{d}).
The autocorrelation function is therefore better captured by an exponential with offset, $r_k = b_\mathrm{offset} \cdot m_\mathrm{offset}^k + c_\mathrm{offset}$.
If the residual of this exponential with offset $R^2_{\mathrm{offset}}$ was smaller than the residual of the MR model $ R^2_{\mathrm{exp}}$ by a factor of two, $H_{\mathrm{offset}} = (2 \cdot  R^2_{\mathrm{offset}} < R^2_{\mathrm{exp}})$, then the data set was rejected. The factor two punishes for the differences in degree of freedom: The residuals of a model with two free parameters (exponential with offset) instead of one (exponential only) can only be smaller. 

Second, ramping of the drive can lead to overestimation of $m$ (Fig. \ref{fig:supp_monkey_nonstationaritites}\textbf{e}).
The comparison of the two models with and without offset introduced above serves as a consistency check able to identify ramping: if the data are captured by a BP, both models should infer identical $\mh$.
Thus, a difference between $\mh_\mathrm{exp}$ and $\mh_\mathrm{offset}$ hints at the invalidity of MR estimation.
Instead of $\mh$, we compared the autocorrelation times $\hat{\tau}_\mathrm{offset} = - \Delta t / \log \hat{m}_\mathrm{offset}$ and $\hat{\tau}_\mathrm{exp}$ obtained from both models, as the logarithmic scaling increases the sensitivity.
If their relative difference was too large, then the data are inconsistent with a BP and MR estimation is invalid: $H_\tau = ( | \tau_{\mathrm{exp}} - \tau_{\mathrm{offset}}| \, / \, \min \lbrace \tau_{\mathrm{exp}}, \tau_{\mathrm{offset}} \rbrace > 2 )$.

Third, when a system changes between different states of activity, e.g. up and down states, the drive rate $\E{h_t}$ may experience sudden jumps.
These can lead to spurious autocorrelation (Fig. \ref{fig:supp_monkey_nonstationaritites}\textbf{f}).
To identify these trends resulting from non-stationary input $h_t$ or from choosing too short data sets, we tested whether the sequence of $r_k$ was fit better by a linear regression  $r_k = q_1 k + q_2$  on the pairs $(k, r_k)$, than by the exponential relation \eqref{eq:mr_residual}. If the residuals $R^2_{\mathrm{lin}}$ were smaller than $ R^2_{\mathrm{exp}}$:  $H_{\mathrm{lin}} = ( R^2_{\mathrm{lin}} < R^2_{\mathrm{exp}})$, data were rejected. 

Apart from non-stationarities, even Poisson activity ($m=0$, $A_{t} = h_t$) with stationary rate may lead to a spurious overestimation of $\mh$ as well:
for \textit{subsampled} branching processes of \textit{finite} duration, the Poisson case and processes close to criticality ($m=1$) can show very similar autocorrelation results, because the sequence of $r_k$ is expected to be absolutely or almost flat, respectively.
Moreover, for $m=0$ any solution on the manifold with $b=0$ minimizes the residuals in Eq. \eqref{eq:mr_residual}. 
Hence, the estimator for $\hat{m}$ may yield any value depending on the initial conditions of the minimization scheme.
To distinguish between $m=0$ and $m>0$, we used the fact that for $m=0$, all slopes $r_k$ are expected to be distributed around zero, $\E{r_k} = 0$.
In contrast, for processes with $m > 0$, all slopes are expected to be larger than zero $\E{r_k} = b \cdot m^k > 0$.
Thus to identify stationary Poisson activity, we tested (using a one-sided t-test) if the  slopes obtained from the data were significantly larger than zero, yielding the $p$-value $p_{\bar{r} \leq 0}$ and the following test (Fig. \ref{fig:supp_monkey_nonstationaritites}\textbf{b}): $H_{\bar{r} \leq 0 } = (p_{\bar{r} \leq 0} \geq 0.1)$.
The choice of the significance level should be guided by the severity of type I or II errors here: if it is set too liberal, Poisson activity may be mistaken for correlated activity, potentially even close-to-critical.
On the other hand, if the significance level is too conservative, activity with long autocorrelation times may be spuriously considered Poissonian under strong subsampling (when $b$ is small and all slopes only slightly differ from zero).
For this study, we chose a significance level of $p_{\bar{r} \leq 0} < 0.1$ in order to not underestimate the risk of large activity cascades.
To confirm candidates for Poisson activity identified through positive $H_{\bar{r} \leq 0 }$, we assured that the $r_k$ did not show a systematic trend, i.e. that linear regression of $r_k$ as a function of $k$ (see $H_{\mathrm{lin}}$ above) yielded slope zero: $H_{q_1 = 0} = (p_{q_1 = 0} \geq 0.05)$.
The according significance level for this two sided test is then given by $p_{q_1 \neq 0} < 0.05$.

We discriminate the following cases in the order indicated in Tab. \ref{tab:supp_consistency}:
 $\mh$ obtained from MR estimation is only valid if none of the tests (except $H_{q_1 = 0}$, which is ignored here) is positive.
A positive result for any of  $H_{\mathrm{offset}}$, $H_\tau$, or $H_{\mathrm{lin}}$ indicates non-stationarities, the data are not explained by a stationary BP, and MR estimation is invalid.
If $H_{\bar{r} \leq 0}$ is positive, the data are potentially consistent with Poisson activity ($m =0 $).
This is only the case if $H_{q_1 = 0}$ is also positive. If otherwise $H_{q_1 = 0}$ is negative, the Poisson hypothesis is also rejected and MR estimation invalid.
This strategy correctly identified the validity of MR estimation for all investigated cases: stationary BPs with $m=0.98$ and $m=0.0$ were accepted, while nonstationary BPs with transient changes, ramping, or sudden jumps of the drive were excluded (Fig. \ref{fig:supp_monkey_nonstationaritites}).

\begin{table}
\centering
\begin{tabular}{ccccccc}
 $H_{\mathrm{offset}}$ & $H_{\tau}$ & $H_{\mathrm{lin}}$ &  $H_{\bar{r} \leq 0 }$ & ($H_{q_1 = 0}$) & interpretation & \\ \cline{1-7} 
 $\times$ & $\times$ & $\times$ & $\times$ & -- &  BP with $m = \hat{m}$ explains data & MR estimation valid\\ \cline{6-7}
 $\checkmark$ & -- & -- & -- &  -- & \multirow{4}{*}{data not explained by BP} & \multirow{4}{*}{MR estimation invalid} \\
 -- & $\checkmark$ & -- & -- & -- &  \\
 -- & -- & $\checkmark$ & -- & -- &  \\
 -- & -- & -- & $\checkmark$ & $\times$ &  \\ \cline{6-7}
 -- & -- & -- & $\checkmark$ & $\checkmark$ & Poisson activity ($m = 0$) explains data & MR estimation valid\\ 

\end{tabular}
\caption{\textbf{Consistency checks for MR estimation.}
In order to assess if the results obtained from MR estimation are consistent with a BP with stationary parameters, we perform five tests (\ref{sec:supp_poisson}). 
We discriminate the following cases in this order:
A BP with $m = \hat{m}$ is only considered to explain the data, if the four tests $H_{\mathrm{offset}}$, $H_\tau$, $H_{\mathrm{lin}}$, and $H_{\bar{r} \leq 0 }$ are negative ($\times$).
If any of $H_{\mathrm{offset}}$, $H_\tau$, or $H_{\mathrm{lin}}$ is positive ($\checkmark$), the data cannot be explained by a BP with any $m$, regardless of the other tests (--), and MR estimation is invalid.
If $H_{\bar{r} \leq 0 }$ is positive, the additional test $H_{q_1 = 0}$ becomes relevant: if it is negative, the data cannot be explained by a BP with any $m$.
If it is also positive, the data are consistent with Poisson activity (BP with $m=0$). 
}
\label{tab:supp_consistency}
\end{table}

\subsection{Variance of the estimates.}
\label{sec:supp_variance}

The distribution of $\mh$ is consistent with a normal distribution $\mathcal{N}(m, \sigma_{\hat{m}}^2)$ centered around the true mean offspring $m$ (Fig. \ref{fig:supp_variance_estimation}\textbf{a}; numerical results).
The variance $\sigma^2_{\mh}$ depends on the branching ratio $m$, the mean activity $\E{A_{t}}$, the length $L$ of the time series, and the sampling fraction $\alpha$.
Each of these factors affects $\sigma^2_{\mh}$ mainly by changing the \emph{effective length} of the time series, i.e. the number of non-zero entries $l = | \lbrace A_{t} \cond A_{t} > 0 \rbrace |$.
Thus, regardless of the actual time series length $L$ or the mean activity $\E{A_{t}}$, the variance scales as a power-law in $l$, $\Var{\hat{m}} \propto l^{-\gamma}$ (Fig. \ref{fig:supp_variance_estimation}\textbf{b}).
The exponent of this power-law depends on $m$. The closer to criticality the process is, the larger the exponent $\gamma$, i.e. the larger the benefit from longer time series length $l$. For $m=0.99$, we found $\gamma \approx 3/2$.
The performance of the estimator is in principle independent of the mean activity: Small $\E{A_{t}}$ only affect the variance of the MR estimator through a potential decrease of $l$.

Similarly, the degree of subsampling only affects the variance of the estimator through a decrease of the effective length of $a_{t}$. While there may be a significant rise in $\sigma^2_{\mh}$ when reducing the sampling fraction $\alpha$, this  increase can be explained by the coincidental decrease in $l$, as the rescaled variance $\sigma^2_{\mh} \cdot l^{\gamma}$ remains within one order of magnitude over four decades of the sampling fraction $\alpha$ (Fig. \ref{fig:supp_variance_estimation}\textbf{c}).

How does the variance change close to the critical transition?
We found that the answer to this question highly depends on the specific choice of the parameters: if $m$ is varied, one can either keep $\E{A_{t}}$ or $h$ constant, not both at the same time.
If the mean activity $\E{A_{t}}$ is fixed by choosing $h = \E{A_{t}} \, (1 - m)$, then the variance of the process scales as $\Var{A_{t}} \propto 1 / (1-m)$ (Theorem \ref{theorem:exp_and_var}).
As $m \rightarrow 1$, the activity will inevitably get into a regime, where bursts of activity ($A_{t} > 0$) are disrupted by intermittent quiescent periods ($A_{t}$), thereby reducing $l$.
In turn, the variance of the estimator increases as detailed before. 

If however, the drive $h$ is kept constant, we found that the variance scales linearly in the distance to criticality $\epsilon = 1 - m$ over at least 5 orders of magnitude of $\epsilon$: $\sigma_{\hat{m}}^2 \propto \epsilon$ (Fig. \ref{fig:supp_variance_estimation}\textbf{d}).
Thus, the variance decreases when approaching criticality, while the relative variance $\sigma_{\hat{m}}^2 / \epsilon$ is constant.
Note, however, that even though the standard deviation also decreases when approaching criticality ($\sigma_{\hat{m}} \propto \sqrt{\epsilon}$), the relative standard deviation increases ($\sigma_{\hat{m}} / \epsilon \propto 1 / \sqrt{\epsilon}$).

For other measures of variation (e.g. quadratic (like the mean squared error MSE) and linear (like the inter-quartile range IQR)), we obtained scaling laws with the same exponents.

\paragraph{Confidence interval estimation.}
We used a model based approach to estimate confidence intervals for both simulation and experimental data (for Figs. \ref{fig:branching_cartoon}\textbf{c},\textbf{d}, \ref{fig:diseases}\textbf{c},\textbf{d}, and \ref{fig:animals}\textbf{d}), because classical bootstrapping methods underestimate the estimator variance by treating all slopes $r_k$ independently, while they are in fact dependent.
We found that our model based approach constructs more conservative and representative confidence intervals. 

For simulations, we simulated $B \in \BbbN$ independent copies of the investigated model and applied MR estimation to each copy, yielding a collection of $B$ independent estimates $\lbrace \mh^{(b)} \rbrace_{b=1}^B$.

For experimental time series $a_{t}$ with length $L$, mean activity $\E{a_{t}}$, and number of sampled units $n$, MR estimation yields an estimate $\mh$.
We then simulated $B$ copies of branching networks $\lbrace A^{(b)}_t \rbrace_{b=1}^B$ (for simulation details see \ref{sec:supp_models}) with $N = 10,000$ units, $m = \hat{m}$ as inferred by MR estimation, and length $L$ and rate $\E{a_{t}}$ to match the data. 
The rate was matched by setting the drive to $h = \E{a_{t}} \, (1 - \mh) \, N / n $. Thereby, after subsampling $n$ units, the mean activity of each resulting time series $a^{(b)}_t$ matched that of the original time series $a_{t}$, $\E{a^{(b)}_t} = \E{a_{t}}$.
This procedure gives $B$ copies of a BN that all match $a_{t}$ in terms of the mean activity, the branching ratio, time series length, and number of sampled units.
Applying MR estimation to these BNs yields a collection of $B$ independent estimates $\lbrace \mh^{(b)} \rbrace_{b=1}^B$. 
For both simulation and experimental data, the distribution of $\mh$ and confidence intervals can be constructed from this collection.

\subsection{Expectation maximization based on Kalman filtering}
\label{sec:supp_kalman}

Kalman filtering is a method to predict the original time series $A_t$ given a measurement $a_t$, defined for AR(1) processes and affine measurement transformation

\begin{align}
A_{t+1} &\, = m \cdot A_t + h_t\nonumber \\
a_t &\, = \alpha \cdot A_t + \beta_t
\end{align}

\noindent
where $h_t$ and $\beta_t$ are independent Gaussian random variables $h_t \sim \mathcal{N}(h, \xi^2)$ and $\beta_t \sim \mathcal{N}(\beta, \zeta^2)$ and $m$ and $\alpha$ constant real numbers.
Assuming that $A_0 \sim \mathcal{N}(A, \psi)$, Kalman filtering infers the original time series $A_t \cond a_t, \mathcal{M}$ given a measured time series $a_t$ and the known model $\mathcal{M} = (m, h, \xi^2, \alpha, \beta, \zeta^2, A, \psi)$.
Based on an iterative expectation maximization algorithm which incorporates Kalman filtering \autocite{Hamilton1994,Shumway1982,Ghahramani1996}, the model parameters $\mathcal{M}$ can be estimated from a time series $a_t$.
We used this algorithm to infer $m$.
Because of the mutual dependence of the model parameters, we also needed to infer $h$, $\xi^2$, $\alpha$, $\beta$, and $\zeta^2$.
In order to reduce the dimensionality of the maximization step, we disregarded $A$ and $\psi$, as the influence of the initial value decreases if the time series gets long.
For initial values, we chose $m=0.5$ in the center of the range of interest for $m$, $h_t = \E{a_t} \cdot (1 - m)$ (see \ref{sec:supp_bps}), $\xi = 0.1 \cdot h_t$, $\alpha = 1$, $\beta = 0$, and $\zeta = 0.1$.
We further chose $A=\E{a_t}$ and $\psi^2 = \Var{a_t}$ for the two model parameters that were not optimized.

We considered two termination criteria for the EM algorithm: 
First, it is recommended to restrict the EM algorithm to 10 -- 20 cycles in order to avoid overfitting, a common problem with likelihood-based fitting methods for multidimensional model parameters. Therefor we considered $\mh$ inferred after 20 EM cycles.
Second, we considered $\mh$ after the results of two subsequent EM cycles did not differ by more than 0.01\%.

We used the publicly available Python implementation of the Kalman EM algorithm, \textit{pykalman}.
All parameters were chosen as detailed above.
The analysis was performed on a computer cluster, and reached runtimes of several days up to projected runtimes of weeks.
In fact, this computational demand was a limiting factor in terms of widespread application.
In contrast, MR estimation terminated within half a second on the same CPUs.

\subsection{Simulations}
\label{sec:supp_models}

\paragraph{Branching process.}

We simulated BPs according to Eq. \eqref{eq:supp_branching_process} in the following way:
Realizations of the random numbers $y_{t,i}$ and $h_t$ describing the number of offsprings, and the drive, were each drawn from a Poisson distribution: $y_{t,i} \sim \mathrm{Poi}(m)$ with mean $m$, and $h_t \sim \mathrm{Poi}(h)$ with mean $h$, respectively.
Here, we used Poisson distributions as they allow for non-trivial offspring distributions with easy control of the branching ratio $m$ by only one parameter.
For the brain, one might assume that each neuron is connected to $k$ postsynaptic neurons, each of which is excited with probability $p$, motivating a binomial offspring distribution with mean $m = k \, p$.
As in cortex $k$ is typically large and $p$ is typically small, the Poisson limit is a reasonable approximation.
For the performance of the MR estimator and the limit behavior of the BP, the particular form of the law $Y$ is not important such that the special choice we made here does not restrict the generality of our results.

The mean rate $\E{A_{t}}$ depends on $m$ and $h$ (Lemma \ref{theorem:exp_and_var}). 
In the simulation we varied $m$ and fixed $\E{A_{t}} = 100$ by adjusting $h$ accordingly if not stated otherwise.
For subsampling the BP, each unit is observed independently with probability $p \leq 1$ .
Then $a_{t}$ is distributed following a binomial distribution $\mathrm{Bin}(A_{t}, p)$, and subsampling is implemented by drawing $a_{t}$ from $A_{t}$ at each time step.
As $\E{a_{t}} = p \, A_{t}$, this implementation of subsampling satisfies the definition of stochastic subsampling with $\alpha = p$, $\beta = 0$.

\paragraph{Branching network.}
In addition to the classical branching process, we also simulated a branching network model (BN) by mapping a branching process \autocite{Harris1963,Haldeman2005} onto a fully connected network of $N=10,000$ neurons. 
An active neuron activated each of its $k$ postsynaptic neurons with probability $p = m / k$.
Here, the activated postsynaptic neurons were drawn randomly without replacement at each step, thereby avoiding that two different active neurons would both activate the same target neuron.
Similar to the BP, the BN is critical for $m = 1$ in the infinite size limit, and subcritical (supercritical) for $m < 1$ ($m > 1$).
As detailed for the BP, $h$ was adjusted to the choice of $m$ to achieve $\E{A_{t}} = 100$, which corresponds to a rate of 0.01 spikes per neuron and time step.
Subsampling \autocite{Priesemann2009} was applied to the model by sampling the activity of $n$ neurons only, which were selected randomly before the simulation, and neglecting the activity of all other neurons.

\paragraph{Self-organized critical model.}

The SOC neural network model we used here is the Bak-Tang-Wiesenfeld (BTW) model \autocite{Bak1987}.
Translated to a neuroscience context, the model consisted of $N=10,000$ ($100 \times 100$) non-leaky integrate and fire neurons.
A neuron $i$ spiked if its membrane voltage $V_i(t)$ reached a threshold $\theta$:

\begin{equation}
\mathrm{If} \, \, V_i(t) > \theta, \, V_i(t+1) = V_i(t) - 4.
\end{equation} 

\noindent
Note that the choice of $\theta$ does not change the activity of the model at all, so we set $\theta= 0$ for convenience.
The model neurons were arranged on a 2D lattice, and each neuron was connected locally to its four nearest neighbors with coupling strength $\alpha_{ij} = \alpha$:

\begin{equation}
V_i(t+1) = V_i(t) + \sum_j \alpha_{ij} \delta(t - T_j) + h_i(t),
\end{equation}

\noindent
where $T_j$ denotes the spike times of neuron $j$, and $h_i(t)$ is the Poisson drive to neuron $i$ with mean rate $h$ as defined for the BP above.
Note that the neurons at the edges and corners of the grid had only 3 and 2 neighbors, respectively.
This model is equivalent to the well-known Bak-Tang-Wiesenfeld model \autocite{Bak1987} if $h \rightarrow 0$ and $\alpha = 1$.
Subsampling \autocite{Priesemann2009} was implemented in the same manner as for the BN. 

\paragraph{Parameter choices.}

If not stated otherwise, simulations were run for $L=10^7$ time steps or until $A_{t}$ exceeded $10^9$, i.e. approximately half of the 32 bit integer range. If not stated otherwise, confidence intervals (\ref{sec:supp_variance}) were estimated from $B=100$ samples, both for simulation and experiments.

In Figs. \ref{fig:branching_cartoon}\textbf{c},\textbf{d}, BNs and the BTW model were simulated with $N=10^4$ units and $\E{A_{t}} = 100$. In Fig. \ref{fig:branching_cartoon}\textbf{e}, BPs were simulated with $m=0.9$ and $\E{A_{t}} = 100$.

In Fig. \ref{fig:animals}\textbf{c}, subcritical and critical BNs with $N=10^4$ and $\E{A_{t}} = 100$ were simulated, and $n=100$ units sampled. Because of the non-stationary, exponential growth in the supercritical case, here BPs were simulated with $h = 0.1$ and units observed with probability $\alpha=0.01$.


\subsection{Epidemiological recordings}
\label{sec:supp_epidemiology}

\paragraph{WHO data on measles worldwide.}
Time series with yearly case reports for measles in 194 different countries are available online from the World Health Organization (WHO) for the years between 1980 and 2014.
MR estimation was applied to these time series.
Because they contain very few data points and potential long-term drifts, we applied the consistency checks detailed above for every country (Tab. \ref{tab:supp_consistency}).
After these checks, 124 out of the 194 surveyed countries were accepted for MR analysis and included in our analysis.
Yearly information on approximate vaccination percentages (measles containing vaccine dose 1, MCV1) for the same countries and time span are also available online from the WHO.

\paragraph{RKI data on norovirus, measles and MRSA in Germany.}
For Germany, the Robert-Koch-Institute (RKI) surveys a range of infectious diseases on a weekly basis, including measles, norovirus, and invasive meticillin-resistant Staphylococcus aureus (MRSA). Case reports are available through their SURVSTAT@RKI server \autocite{RKI2016}. Because of possible changes in report policies in the beginning of surveillance, we omitted the data from the first 6 months of each recording. Moreover, we omitted the incomplete week on the turn of the year, thus evaluating 52 full weeks in each year.

The MRSA recording showed a slow, small variation in the case reports that can be attributed to slow changes in the drive rates. To compensate for these slow drifts, we corrected the time series by subtracting a moving average over 3 years (156 weeks). We then applied MR estimation to the obtained time series.
The recordings for measles and norovirus showed strong seasonal fluctuations of the case reports, resulting in a baseline oscillation of the autocorrelation function. We therefore used a modified model

\begin{equation}
r_k = b \cdot m^k + c \cdot \cos (2 \pi k / T)
\label{eq:mr_seasonal}
\end{equation}

\noindent
with a fixed period of $T = \SI{52}{weeks}$, and estimated $\hat{m}$, $\hat{b}$, and $\hat{c}$ from minimizing the residual of this modified equation.

In order to obtain the naive estimates using the conventional linear regression estimator $\hat{m}_\mathrm{C} = \hat{r}_1$, we used the following correction for seasonal fluctuations.
Each incidence count $a_t$ was normalized by the incidence counts from the same week, averaged over all years of recording ($\bar{a}_w = \E[y]{a_{w + 52 \cdot y}}$ with the average taken over the years $y$ for any week $w=1,\ldots,52$), yielding the deseasonalized time series $a^\prime_t = a_t / \bar{a}_{t \, \mathrm{mod} \, 52}$.
Linear regression was performed on this time series $a^\prime_t$.

For Fig. \ref{fig:diseases}\textbf{d}, subsampling was applied to the original time series assuming that every infection is diagnosed and reported  with a probability $\alpha$, yielding the binomial subsampling described in \ref{sec:supp_subsampling}. MR estimates were obtained from this subsampled time series according to Eq. \eqref{eq:mr_seasonal}, for the conventional estimator the subsampled time series was processed as described above.

\subsection{Animal experiments}
\label{sec:supp_animals}
We evaluated spike population dynamics from recordings in rats, cats and monkeys.
 The rat experimental protocols were approved by the Institutional Animal Care and Use Committee of Rutgers University \autocite{Mizuseki2009,Mizuseki2009a}.
 The cat experiments were performed in accordance with guidelines established by the Canadian Council for Animal Care \autocite{Blanche2009}.
The monkey experiments were performed according to the German Law for the Protection of Experimental Animals, and were approved by the Regierungspr\"asidium Darmstadt.
The procedures also conformed to the regulations issued by the NIH and the Society for Neuroscience.
The spike recordings from the rats and the cats were obtained from the NSF-founded CRCNS data sharing website \autocite{Blanche2006,Blanche2009,Mizuseki2009,Mizuseki2009a}. 

In rats the spikes were recorded in CA1 of the right dorsal hippocampus during an open field task.
We used the first two data sets of each recording group (ec013.527, ec013.528, ec014.277, ec014.333, ec015.041, ec015.047, ec016.397, ec016.430).
The data-sets provided sorted spikes from 4 shanks (ec013) or 8 shanks (ec014, ec015, ec016), with 31 (ec013), 64 (ec014, ec015) or 55 (ec016) channels. We used both, spikes of single and multi units, because knowledge about the identity and the precise number of neurons is not required for the MR estimator.
More details on the experimental procedure and the data-sets proper can be found in \cite{Mizuseki2009,Mizuseki2009a}.

For the spikes from the cat, neural data were recorded by Tim Blanche in the laboratory of Nicholas Swindale, University of British Columbia \autocite{Blanche2009}.
We used the data set pvc3, i.e. recordings in area 18 which contain 50 sorted single units \autocite{Blanche2006}.
We used that part of the experiment in which no stimuli were presented, i.e., the spikes reflected spontaneous activity in the visual cortex of the anesthetized cat. Because of potential non-stationarities at the beginning and end of the recording, we omitted data before \SI{25}{s} and after \SI{320}{s} of recording.
Details on the experimental procedures and the data proper can be found in \cite{Blanche2009,Blanche2006}.

The monkey data are the same as in \cite{Pipa2009, Priesemann2014}. In these experiments, spikes were recorded simultaneously from up to 16 single-ended micro-electrodes ($ \diameter  = 80 \, \mu  \mathrm{m}$) or tetrodes ($ \diameter  = 96 \, \mu  \mathrm{m}$) in lateral prefrontal cortex of three trained macaque monkeys (M1: 6 kg \female ; M2: 12 kg \male ; M3: 8 kg \female ).
The electrodes had impedances between 0.2 and $1.2 \, \mathrm{M} \Omega$ at 1 kHz, and were arranged in a square grid with inter electrode distances of either 0.5 or 1.0 mm.
The monkeys performed a visual short term memory task. The task and the experimental procedure is detailed in \cite{Pipa2009}.
We analyzed spike data from 12 experimental sessions comprising almost 12.000 trials (M1: 4 sessions; M2: 5 sessions; M3: 3 sessions).
6 out of 12 sessions were recorded with tetrodes.
Spike sorting on the tetrode data was performed using a Bayesian optimal template matching approach as described in \cite{Franke2010} using the “Spyke Viewer” software \autocite{Propper2013}. 
On the single electrode data, spikes were sorted with a multi-dimensional PCA method (Smart Spike Sorter by Nan-Hui Chen).

\paragraph{Analysis.}
For each recording, we collapsed the spike times of all recorded neurons into one single train of population spike counts $a_{t}$, where $a_{t}$ denotes how many neurons spiked in the $t^{th}$ time bin $\Delta t$.
We used $\Delta t = \SI{4}{ms}$, reflecting the propagation time of spikes from one neuron to the next.
Note that $m$ scales with the bin size (bs) as $m(\mathrm{bs}=k \Delta t) = m(\mathrm{bs}=\Delta t)^k$, while the corresponding autocorrelation times are invariant under bin size changes.
For Figs. \ref{fig:animals}\textbf{b} and \ref{fig:supp_cat_single_electrodes}, we investigated single neuron activity by applying similar binning to the spike times of each neuron individually.

From these time series, we estimated $\hat{m}$ using the MR estimator with $k_\mathrm{max} = 2500$ (corresponding to \SI{10}{s}) for the rat recordings, $k_\mathrm{max} = 150$ (\SI{600}{ms}) for the cat recording, and $k_\mathrm{max} = 500$ (\SI{2000}{ms}) for the monkey recordings, assuring that $k_\mathrm{max}$ was always in the order of multiple autocorrelation times. 
Experiments were excluded if the tests according to \ref{sec:supp_poisson} detected potential nonstationarities.


\begin{scriptsize}

\end{scriptsize}

\newpage

\begin{figure}
\includegraphics[width=\textwidth]{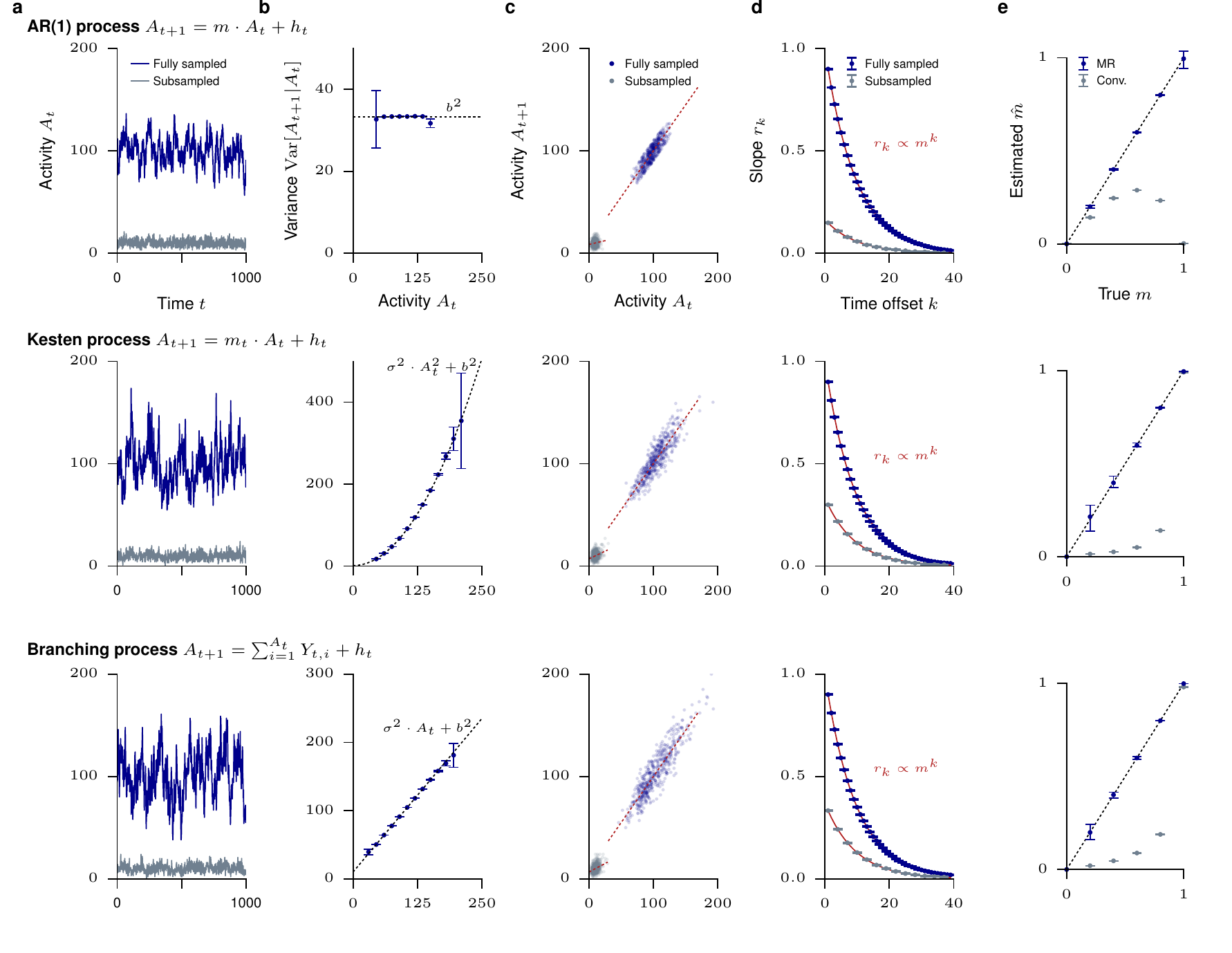}
\caption{\textbf{MR estimation for PARs.} Although derived for branching processes (BPs), we conjectured that MR estimation is applicable to any process with a first order autoregressive representation (PAR). We here show exemplary results for three different classes of PARs:
In AR(1) processes, additive noise $h_t$ is drawn independently at each time step. Here, we considered a uniform distribution $h_t \sim \mathcal{U}(0, 2h)$.
In a Kesten process, additive and multiplicative noise is drawn at each time step, both $m_t$ and $h_t$ being i.i.d. for all $t$. Here, $m_t \sim \mathcal{N}(m, \sigma^2)$ with $\sigma = m / 10$ and $h_t \sim \mathcal{N}(h, b^2)$ with $b = h / 10$ are normally distributed.
In a BP, each unit $i$ at time $t$ generates $Y_{t,i}$ offspring, which are i.i.d. for all $t$ and $i$.
In addition, a random number $h_t$ of units are introduced at each time step. Here, $Y_{t,i} \sim \mathrm{Poi}(m)$ and $h_t \sim \mathrm{Poi}(h)$ are Poisson distributed, $\sigma^2$ and $b^2$ denote the variances of $Y_{t,i}$ and $h_t$ respectively.
All three processes satisfy the first-order statistical recursion relation $\E{A_{t+1} \cond A_t} = m A(t) + h$ (Eq. \eqref{eq:PAR_equation}).
Parameters are chosen such that for all simulations the average activity is identical, $\E{A_t} = 100$.
\textbf{a}. Fully sampled and subsampled (binomial subsampling $a_t \sim \mathrm{Bin}(A_t, \alpha)$ with $\alpha = 1/10$) time series are shown for $m=0.9$ and $h=10$.
\textbf{b}. The three classes show the same first-order statistics according to Eq. \eqref{eq:PAR_equation}. However, their second order statistics $\Var{A_{t+1} \cond A_t}$ differ as indicated.
\textbf{c}. Conventional linear regression underestimates $\mh$ for all three processes under subsampling.
\textbf{d}. MR estimation is applicable to all three processes under full sampling and subsampling, i.e. $r_k \propto m^k$ holds.
\textbf{e}. While MR estimation returns consistent estimates of $m$ even under subsampling, the conventional estimator underestimates $\mh$ for all three processes.
}
\label{fig:supp_ar_processes}
\end{figure}

\pagebreak

\begin{figure}
\centering
\includegraphics[width=140mm]{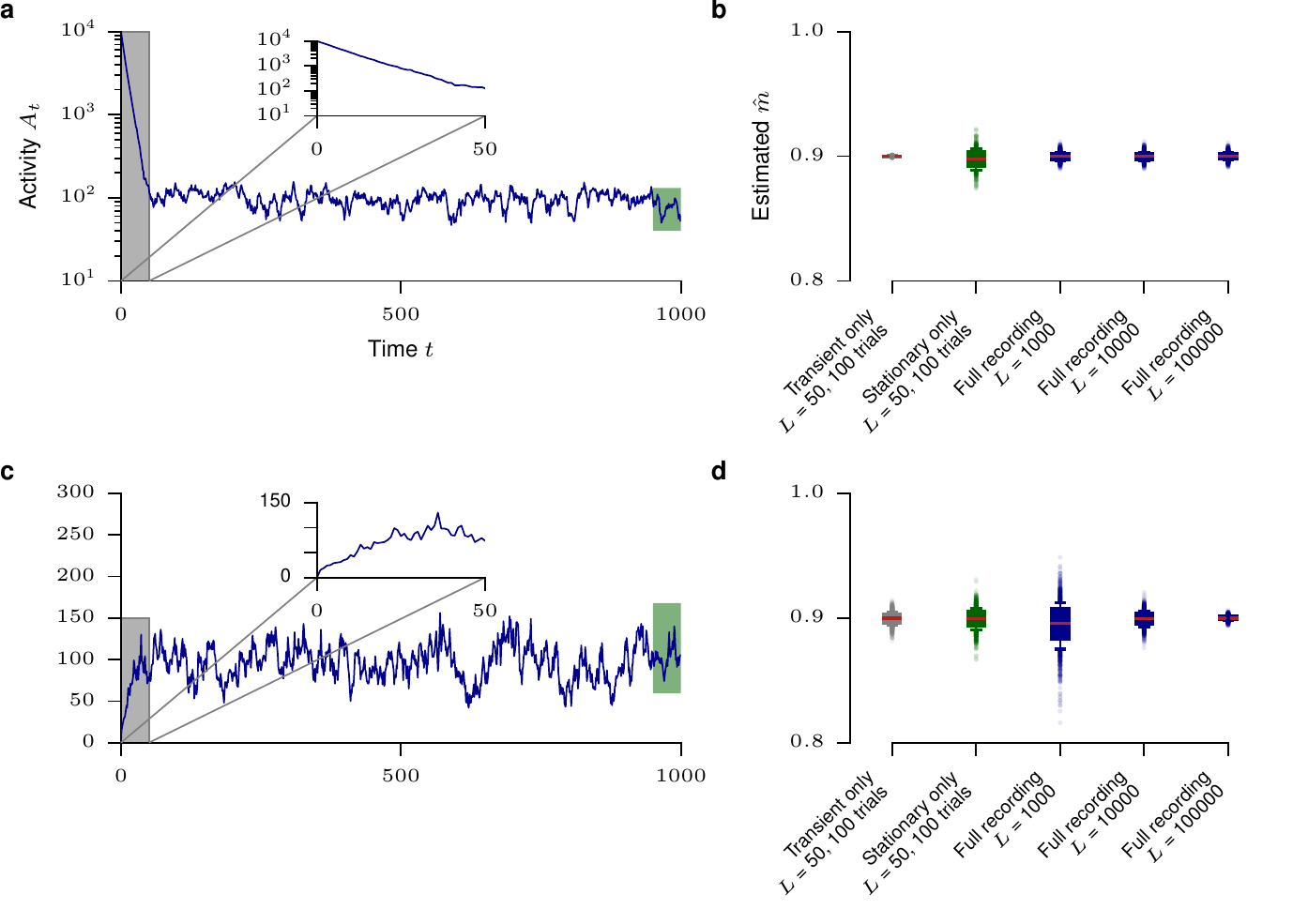}
\caption{\textbf{MR estimation with transients.} A branching process (BP) with $m=0.9$ and expected activity $\E{A_t} = 100$ is started far from the stationary distribution, namely with $A_0 = 10,000$ (top) or $A_0 = 0$ (bottom). Using MR estimation, $\mh$ is inferred from: (i) only the first 50 data points of 100 independent trials, i.e. only transient parts of the activity in each trial (gray); (ii) 50 data points of 100 independent trials after the activity was allowed to relaxate to the stationary distribution in each trial (green); (iii) from one single trial comprising both transient and stationary parts, using $10^3$, $10^4$, or $10^5$ time steps (blue). 
\textbf{a}, \textbf{c}. Activity $A_t$ of one single trial of $10^3$ time steps as a function of time $t$. Insets show magnified transient period where $A_t$ converges to the stationary distribution. Shaded areas indicate transient (gray) and stationary (green) parts taken into account for estimates (i) and (ii) respectively.
\textbf{b}, \textbf{d}. Boxplots (derived from 1000 independent realizations) for the result $\mh$ of MR estimation, based on the data specified above.
}
\label{fig:supp_consistency_transients}
\end{figure}

\pagebreak

\begin{figure}
\includegraphics[width=\textwidth]{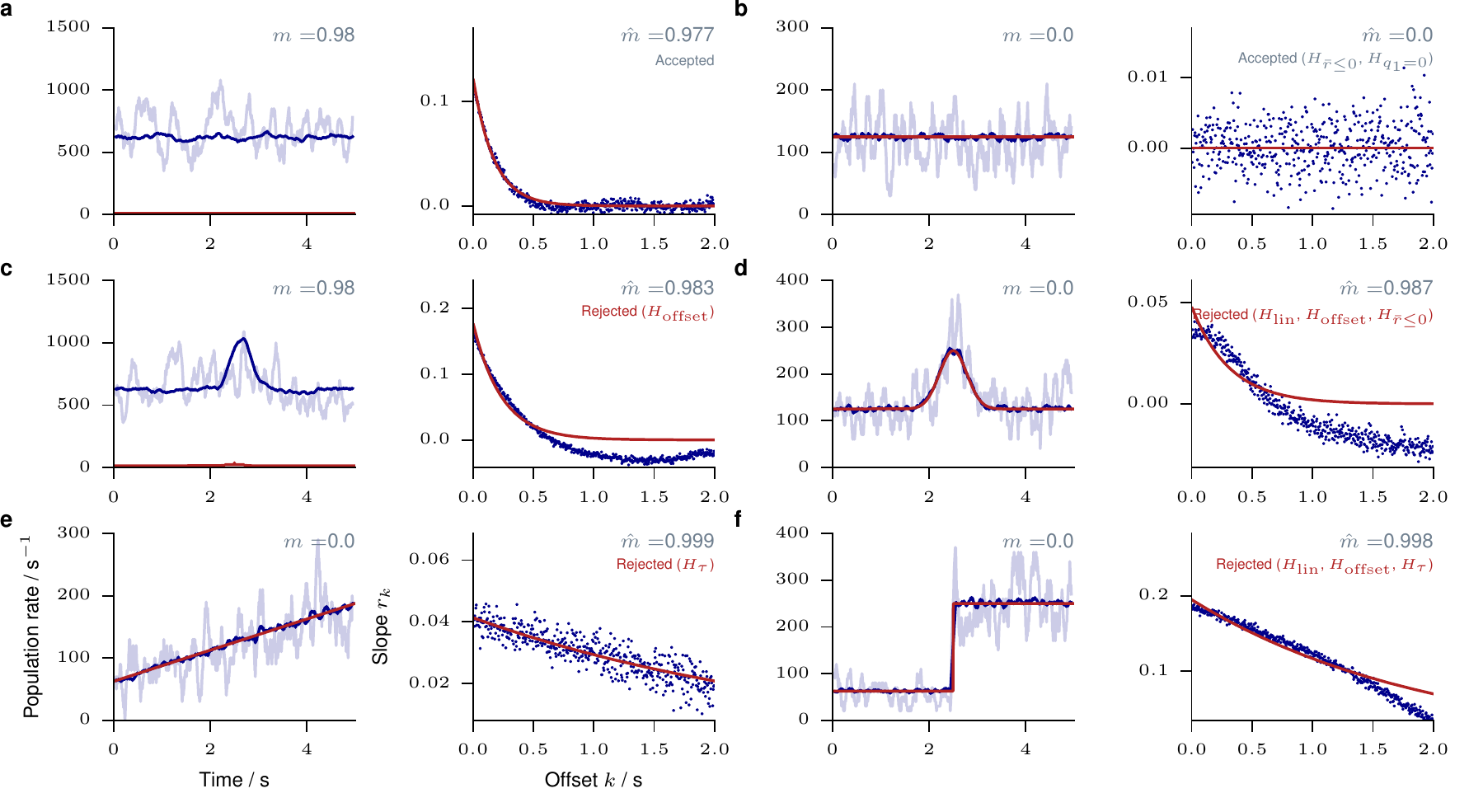}
\caption{\textbf{Excluding nonstationary data.} Each left panels shows the time series $a_t$ of the activity from one single trial (light blue) and averaged activity from 100 trials (dark blue), recorded
 from $n=50$ out of $N=10^4$ neurons. Each right panels shows the corresponding MR estimation from one single trial.
 We investigated the following, generic cases for the temporal evolution of the drive rate $\E{h_t}$:
\textbf{a}, \textbf{b}. The drive is stationary ($\E{h_t}$ identical for all $t$, red), so are the mean rates $\E{a_t}$. 
\textbf{c}, \textbf{d}. The drive exhibits a transient increase centered around half of the simulation time. The mean rate $\E{a_t}$ is therefore also time-dependent and follows the temporal evolution of $\E{h_t}$.
\textbf{e}. The drive shows a linear increase over the simulation.
\textbf{f}. The drive exhibits a step function after half the simulation.
Nonstationarities (\textbf{c} -- \textbf{f}) typically lead to an overestimation of $\mh$, which is particularly severe if the underlying dynamics is Poissonian ($m=0$).
The tests defined in \ref{sec:supp_poisson} (see Tab. \ref{tab:supp_consistency}) were able to exclude time series where the investigated nonstationarities were present, while accepting the stationary cases \textbf{a}, \textbf{b}.
}
\label{fig:supp_monkey_nonstationaritites}
\end{figure}

\newpage

\begin{figure}
\includegraphics[width=0.48\textwidth]{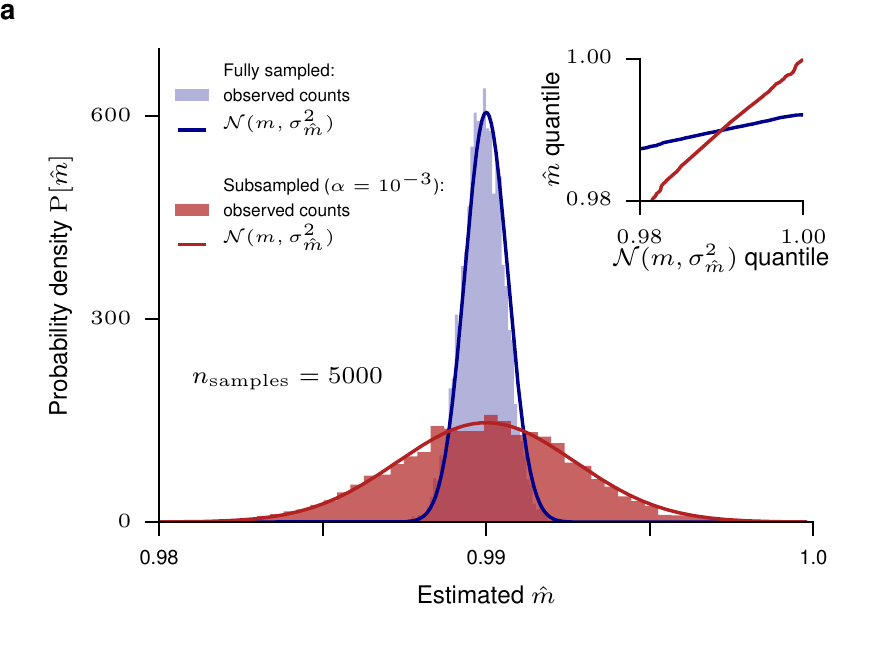}
\includegraphics[width=0.48\textwidth]{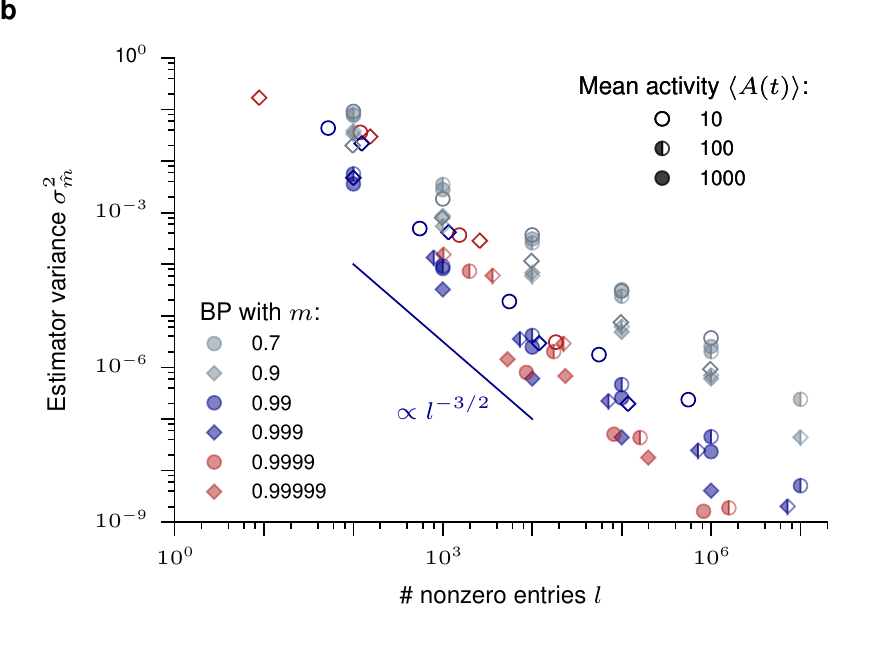}\\
\includegraphics[width=0.48\textwidth]{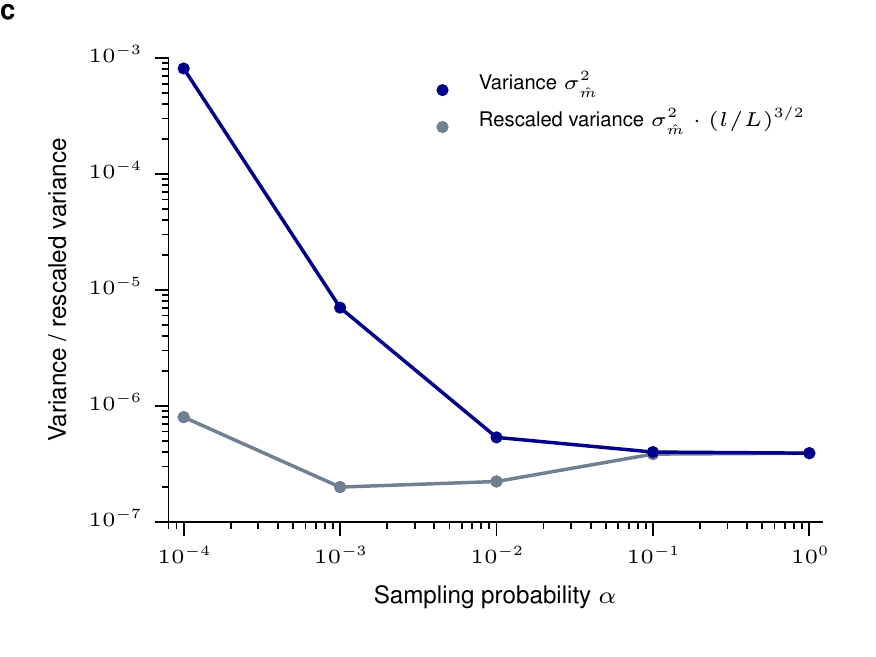}
\includegraphics[width=0.48\textwidth]{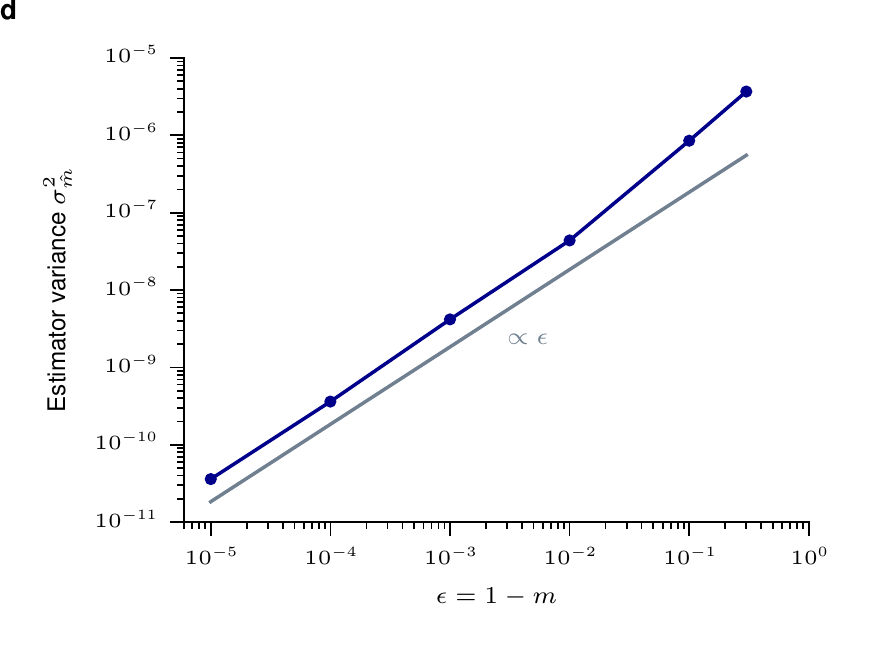}

\caption{\textbf{Variance of the MR estimates.}
This figure shows numerical result for the distribution and variability of the estimate $\mh$ as a function of multiple parameters.
\textbf{a}. Distribution of the estimate $\hat{m}$, estimated from 5000 independent copies of a branching process (BP) with $m = 0.99$, $\E{A_t}= 100$ and length $L=10^5$: normalized histograms of the probability of estimating $\hat{m}$ for full sampling (blue) and binomial subsampling with $\alpha=0.001$ (red), together with normal distributions $\mathcal{N}(m, \hat{\sigma}^2_{\hat{m}})$. Inset: $Q$-$Q$-plot for the quantiles of $\mathcal{N}(m, \hat{\sigma}^2_{\hat{m}})$ and the quantiles of the estimated $\hat{m}$ under both samplings. The estimated $\hat{m}$ are found to be distributed normally in both cases (fully sampled: $r^2 = 0.9995$, subsampled: $r^2 = 0.998$).
\textbf{b}. The variance $\sigma^2_{\mh}$ of the estimate $\mh$ is estimated from 100 independent copies of a BP.
Results for different $m$, mean activities $\E{A_t}$ and time series lengths $L$ are plotted as a function of the  effective time series length $l = | \lbrace A_t \cond A_t > 0 \rbrace |$, the number of nonzero entries.
For any given $m$, the variance of $\hat{m}$ shows algebraic scaling $\sigma^2_{\hat{\epsilon}} \propto l^\gamma$.
The exponent of this scaling depends on $m$, with higher $\gamma$ the closer $m$ is to unity.
Hence, the benefit from longer time series is larger the closer a system is to criticality. 
Importantly, the variance does not directly depend on the mean activity $\E{A_t}$, this number only influences the accuracy of MR estimation via the potential change in $l$.
\textbf{c}. The variance of the estimate $\mh$ is estimated from 100 independent copies of a BP with $m = 0.99$, $\E{A_t}= 100$, and $L=10^5$ and plotted as a function of the sampling probability $\alpha$ under binomial subsampling. While the variance appears to increase dramatically under stronger subsampling, this increase can be attributed to the according decrease of the effective time series length $l$.
After rescaling by $(l/L)^{3/2}$ (cf. panel \textbf{b}), the rescaled variance remains within one order of magnitude over four orders of magnitude in $\alpha$.
Hence, the accuracy of the estimator is not directly influenced by the degree of subsampling.
\textbf{d}. The variance $\sigma^2_{\hat{m}}$ is estimated from 100 independent copies of a BP with $m = 0.99$, $h = 1$, and $L=10^5$ and plotted as function of the distance to criticality $\epsilon = 1 - m$.
The variance is found numerically to scale as $\sigma^2_{\mh} \propto \epsilon$, hence the standard deviation scales as $\sigma_{\mh} \propto \sqrt{\epsilon}$.
Similar scaling results were found for other linear (like the interquartile range) and quadratic (like the mean squared error) measures of variation.
}
\label{fig:supp_variance_estimation}
\end{figure}

\newpage

\begin{figure}
\includegraphics[width=\textwidth]{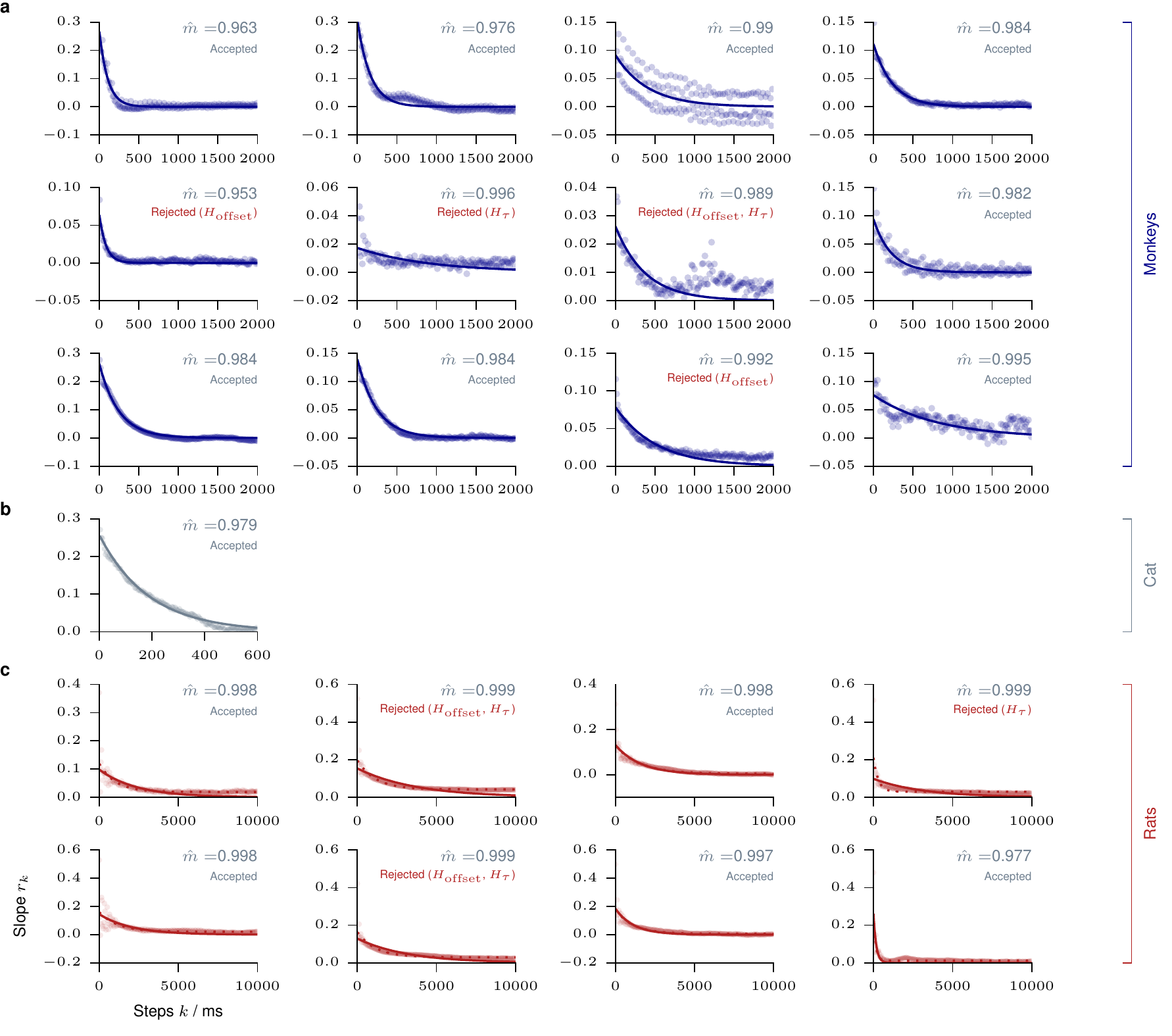}
\caption{\textbf{MR estimation for individual animals.} MR estimation is shown for every individual animal (see \ref{sec:supp_animals}). The consistency checks are detailed in the \ref{sec:supp_poisson} (see Tab. \ref{tab:supp_consistency}). \textbf{a.} Data from monkey prefrontal cortex during an working memory task. The third panel shows a oscillation of $r_k$ with a frequency of 50 Hz, corresponding to measurement corruption due to power supply frequency. \textbf{b.} Data from anesthetized cat primary visual cortex. \textbf{c.} Data from rat hippocampus during a foreaging task. In addition to a slow exponential decay, the slopes $r_k$ show the $\vartheta$-oscillations of 6 -- 10 Hz present in hippocampus. Dashed lines indicate results for an exponential model with offset, solid lines results for the model without offset (compare \ref{sec:supp_poisson}).}
\label{fig:supp_animal_data}
\end{figure}

\newpage

\begin{figure}
\includegraphics[width=\textwidth]{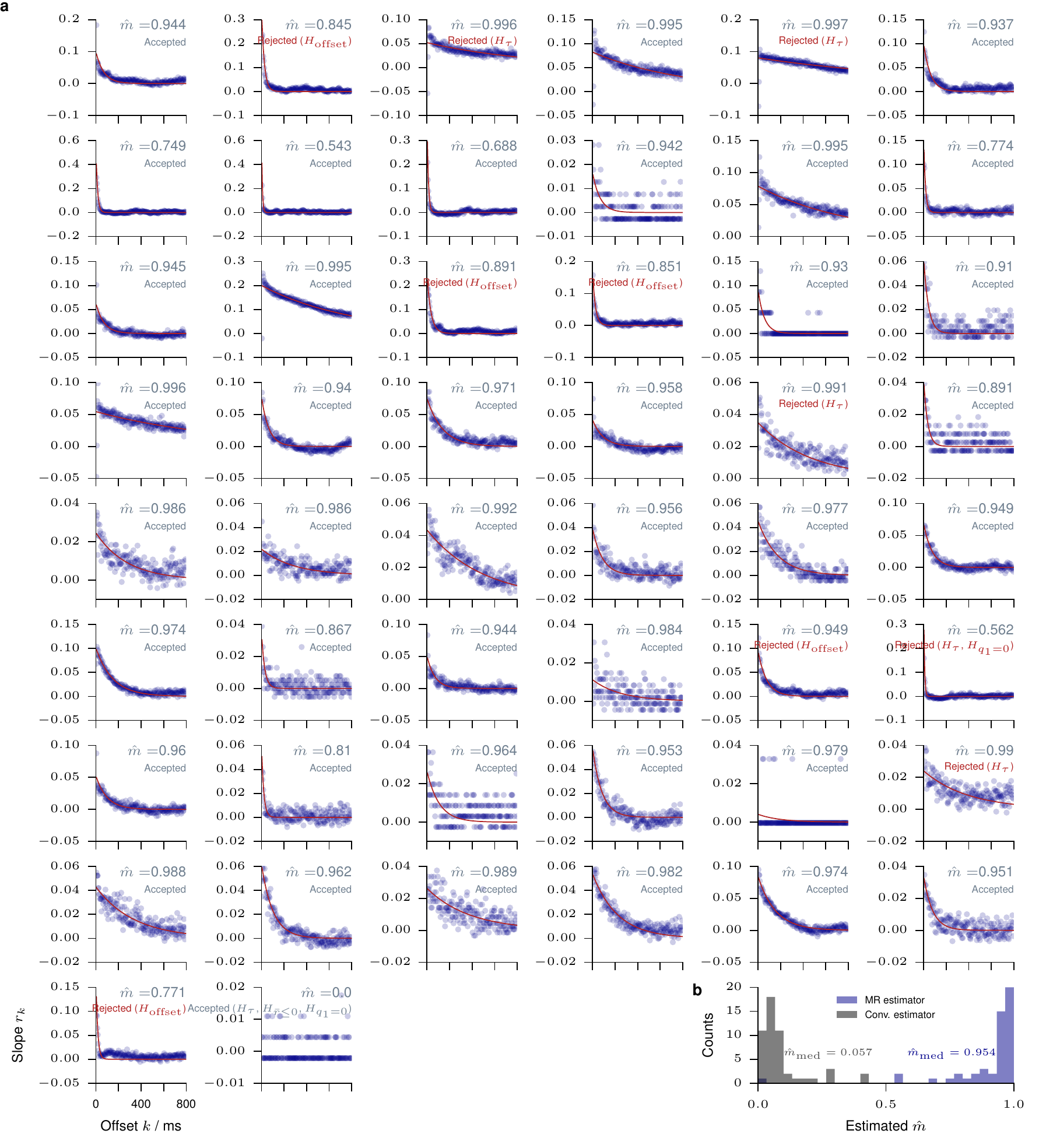}
\caption{\textbf{MR estimation from single neuron activity (cat).} MR estimation is used to estimate $\hat{m}$ from the activity $a_t$ of a single neurons in cat visual cortex. \textbf{a.} Each panel shows MR estimation for one of the 50 recorded neurons. Autocorrelations decay rapidly in some neurons, but long-term correlations are present in the activity of most neurons. The consistency checks are detailed in \ref{sec:supp_poisson} (see Tab. \ref{tab:supp_consistency}). \textbf{b.} Histogram of the single neuron branching ratios $\hat{m}$, inferred with the conventional estimator and using MR estimation. The difference between these estimates demonstrates the subsampling bias of the conventional estimator, and how it is overcome by MR estimation.}
\label{fig:supp_cat_single_electrodes}
\end{figure}

\begin{figure}
\includegraphics[width=\textwidth]{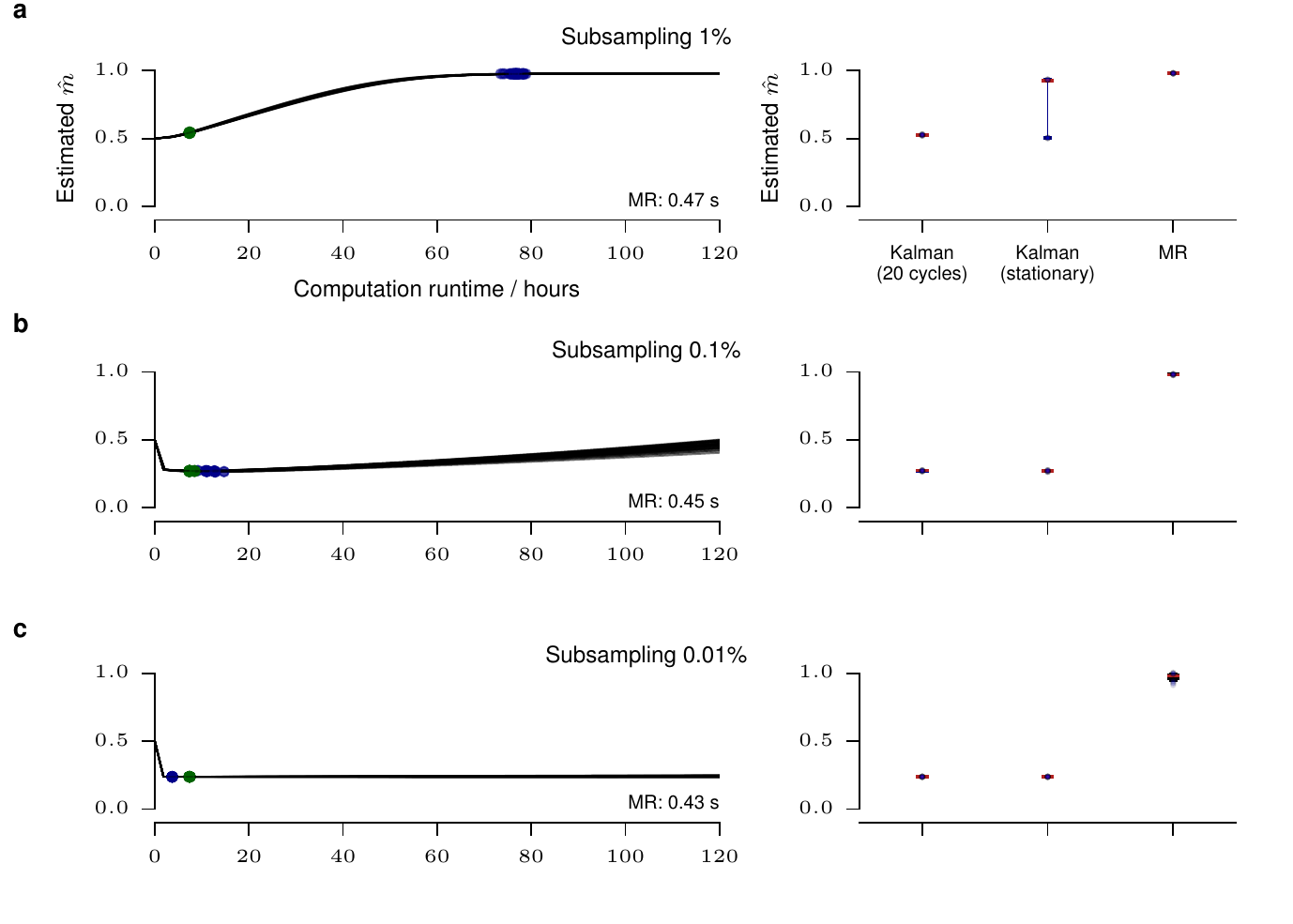}
\caption{\textbf{Kalman EM estimation.} Expectation maximization (EM) based on Kalman filtering and MR estimation are used to infer $\mh$ from BPs with $m=0.99$ and different degrees of subsampling. 
Left column: inferred $\mh$ as a function of the EM runtime for 100 independent copies of the BP.
The EM algorithm is terminated after 20 cycles (green dots) or after the inferred $\mh$ changed only marginally (blue dots, see \ref{sec:supp_kalman}).
The median runtime of MR estimation for the same BPs is also indicated.
Right column: estimated $\mh$ for all three methods.
\textbf{a}. Under 1\% subsampling, the EM algorithm converged after runtimes of about \SI{80}{h}, compared to \SI{0.43}{s} for MR estimation.
\textbf{b}. Under 0.1\% subsampling, $\mh$ inferred by the EM algorithm reaches a steady state after \SI{10}{h}, but is severely biased.
The slow rise of $\mh$ might lead to a convergance to the proper $m$ after several weeks of projected runtime (ignoring common termination criteria).
\textbf{c}. Under 0.01\% subsampling, $\mh$ inferred by the EM algorithm converge to a biased value.
In contrast, MR estimation returns a correct $\mh$ in all three cases, and outperforms the EM algorithm by a factor of $10^5$ to $10^6$ in terms of the runtime.}
\label{fig:supp_kalman}
\end{figure}




\end{document}